\documentclass{article}

\usepackage[a4paper,left=25mm,top=20mm,right=20mm,bottom=15mm]{geometry}

\usepackage{amsmath, amsfonts, amssymb, amsthm}
\usepackage{algorithmic,algorithm}

\providecommand{\keywords}[1]{\textbf{\textit{Keywords: }} #1}

\usepackage{mathtools}

\usepackage{caption}
\usepackage{subcaption}

\usepackage{subfig}

\newcommand{\overlap}{\between}
\newcommand{\mc}{\mathcal}

\newcommand{\lca}{\ensuremath{\operatorname{lca}}}

\newcommand{\Gen}{\ensuremath{\mathbb{G}}}

\newcommand{\Virr}{\ensuremath{V^{\times}_\mathrm{irr}}}
\newcommand{\Vgirr}{\ensuremath{V^{\times}_{g,\mathrm{irr}}}}

\newcommand{\Xirr}{\ensuremath{X^{\times}_\mathrm{irr}}}

\newcommand{\M}{\ensuremath{\mathbb{M}}}
\newcommand{\Mmax}{\ensuremath{\mathbb{M}_{\max}}}
\newcommand{\Ms}{\ensuremath{\mathbb{M}_{\mathrm{str}}}}
\newcommand{\rev}{\ensuremath{\mathrm{rev}}}

\newcommand{\ora}{\overrightarrow}
\newcommand{\vp}{\varphi}

\DeclareFixedFont{\textnew}{OT1}{pzc}{bx}{n}{10pt}
\DeclareFixedFont{\textnewsmall}{OT1}{pzc}{bx}{n}{10.5pt}
\DeclareFixedFont{\textnewlarge}{OT1}{pzc}{bx}{n}{15pt}
\newcommand{\unp}{{\textnew{unp}}}
\newcommand{\captionunp}{{\textnewsmall{unp}}}
\newcommand{\secunp}{{\textnewlarge{unp}}}

\newtheorem{theorem}{Theorem}[section]
\newtheorem{lemma}{Lemma}[section]
\newtheorem{proposition}{Proposition}[section]
\newtheorem{corollary}{Corollary}
\newtheorem{remark}{Remark}
\newtheorem{problem}{Problem}

\newtheorem{definition}{Definition}

\makeindex             

\begin{document}

\title{The Mathematics of Xenology: Di-cographs, Symbolic Ultrametrics, 2-structures and 
  Tree-representable Systems of Binary Relations}

\author{Marc Hellmuth\\ University of Greifswald \\ 
			Dpt.\ of Mathematics and Computer Science \\ Walther-
           Rathenau-Strasse 47,  D-17487 Greifswald, Germany,  \\and \\
				Saarland University, Center for Bioinformatics \\ Building E 2.1, 
				P.O.\ Box 151150,  D-66041 Saarbr{\"u}cken, Germany \\
			 Email: \texttt{mhellmuth@mailbox.org} 
				\and  
		Peter F.\ Stadler \\
        Bioinformatics Group, Department of Computer Science; and\\
        Interdisciplinary Center of Bioinformatics, University of Leipzig, \\
        H{\"a}rtelstra{\ss}e 16-18, D-04107 Leipzig, and  \\
        Max-Planck-Institute for Mathematics in the Sciences, \\
        Inselstra{\ss}e 22, D-04103 Leipzig, and \\
        Inst.\ f.\ Theoretical Chemistry, University of Vienna, \\
        W{\"a}hringerstra{\ss}e 17, A-1090 Wien, Austria, and \\
        Santa Fe Institute, 1399 Hyde Park Rd., Santa Fe,\\
        NM 87501, USA, \\
        Email: \texttt{studla@bioinf.uni-leipzig.de}
		\and
		Nicolas Wieseke \\
	Leipzig University\\ Parallel Computing and Complex Systems Group \\ Dpt.\ of 
     Computer Science \\ Augustusplatz 10, D-04109 Leipzig, Germany \\
    Email: \texttt	{wieseke@informatik.uni-leipzig.de}
	}

\date{\ }

\maketitle

\abstract{ 
The concepts of orthology, paralogy, and xenology play a key role in
  molecular evolution. Orthology and paralogy distinguish whether a pair of
  genes originated by speciation or duplication. The corresponding binary
  relations on a set of genes form complementary cographs. Allowing
  more than two types of ancestral event types leads to symmetric symbolic
  ultrametrics. Horizontal gene transfer, which leads to xenologous gene
  pairs, however, is inherent asymmetric since one offspring copy ``jumps''
  into another genome, while the other continues to be inherited
  vertically. We therefore explore here the mathematical structure of the
  non-symmetric generalization of symbolic ultrametrics. Our main results
  tie non-symmetric ultrametrics together with di-cographs (the directed
  generalization of cographs), so-called uniformly non-prime 2-structures,
  and hierarchical structures on the set of strong modules. This yields a
  characterization of relation structures that can be explained in terms of
  trees and types of ancestral events. This framework accomodates 
  a horizontal-transfer relation in terms of an ancestral event and
  thus, is slightly different from the the most commonly used
  definition of xenology.
}

\smallskip
\noindent

\sloppy

\section{Introduction}

The current flood of genome sequencing data poses new challenges for
comparative genomics and phylogenetics. An important topic in this context
is the reconstruction of large families of homologous proteins, RNAs, and
other genetic elements. The distinction between orthologs, paralogs, and
xenologs is a key step in any research program of this type. The
distinction between orthologous and paralogous gene pairs dates back to the
1970s:  pairs of genes whose last common ancestor in the ``gene tree''
corresponds to a speciation are orthologs; if the last common ancestor was
a duplication event, the genes are paralogs \cite{Fitch:70}. The importance
of this distinction is two-fold: first it is informative in genome
annotation. Orthologs usually fulfill corresponding functions in related
organism. Paralogs, in contrast, are expected to have similar but distinct
functions \cite{Koonin:05}. Secondly, the orthology (or paralogy) relation
conveys information about the events corresponding to internal nodes of the
gene tree \cite{HernandezRosales:12a} and about the underlying species tree
\cite{Hellmuth:13a,Hellmuth:15a}. 

Based on a theory of symbolic ultrametrics \cite{BD98} is was shown in
\cite{Hellmuth:13a} that the orthology and paralogy relations are
necessarily complementary cographs provided the genetic repertoire evolved
only by means of speciation, gene duplication, and gene loss. However,
horizontal gene transfer (HGT), i.e., the incorporation of genes or other
DNA elements from a source different than the parent(s), cannot be
neglected under many circumstances. In fact, HGT plays an important role
not only in the evolution of procaryotes \cite{Koonin:01} but also in
eukaryotes \cite{Keeling:08}. This begs the question whether the
combinatorial theory of orthology/paralogy can be extended to incorporate
xenologs, i.e., pairs of genes that are separated in the gene tree by
HGT events. 

In contrast to orthology and paralogy, the definition of xenology is less
well established and by no means consistent in the biological
literature. The most commonly used definition stipulates that two genes are
\emph{xenologs} if their history since their common ancestor involves
horizontal transfer of at least one of them \cite{Fitch:00,Jensen:01}. In
this setting the HGT event itself is treated as gene duplication
event. Every homolog is still either ortholog or a paralog. Both
orthologs and paralogs may at them same time be xenologs \cite{Jensen:01}.

\begin{figure}[tbp]
  \begin{minipage}{.45\linewidth}
    \begin{itemize} 
    \item[] \scriptsize $R_{o} = \{dv\mid v\in \Gen\setminus\{d\}\}\cup
       \{ab_1, ac_2, b_1c_2, b_2c_2\}$
    \item[] $R_{x} = \{(c_1,b_3)\}$
    \item[] $R_{p} =\Gen^{\times}_\mathrm{irr}  \setminus 
       (R_o\cup R_x\cup \{(b_3,c_1)\})$ 
    \item[] $xy\in R_{\star}$ means that  $(x,y)(y,x)\in R_{\star}$, 
      with $\star \in \{o,p\}$ 
    \end{itemize}

    \vspace*{1mm}
    \centering
    \includegraphics[scale = 0.35]{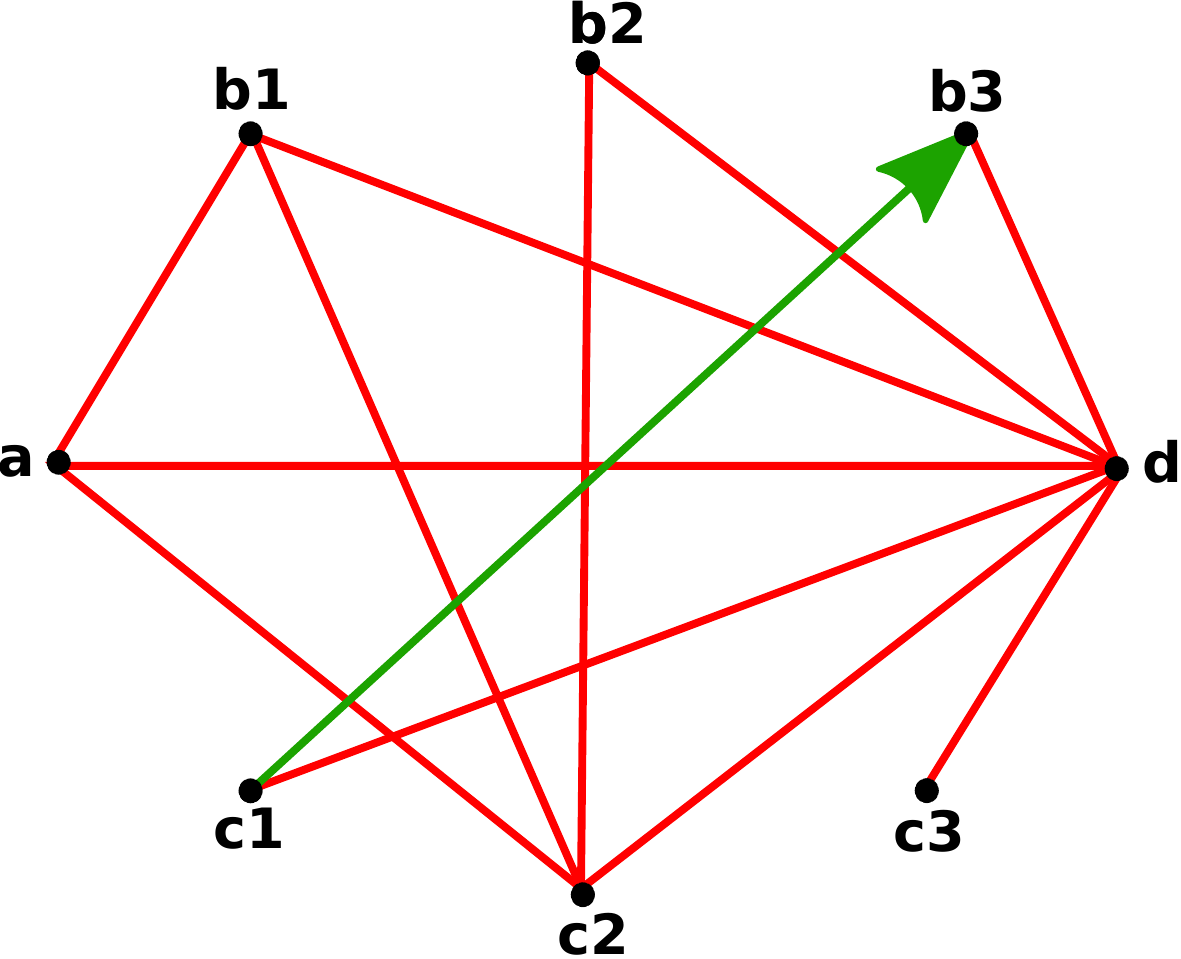}
    \\[0.5cm]
  \end{minipage}
  \begin{minipage}{.4\linewidth}
    \includegraphics[scale = 0.45]{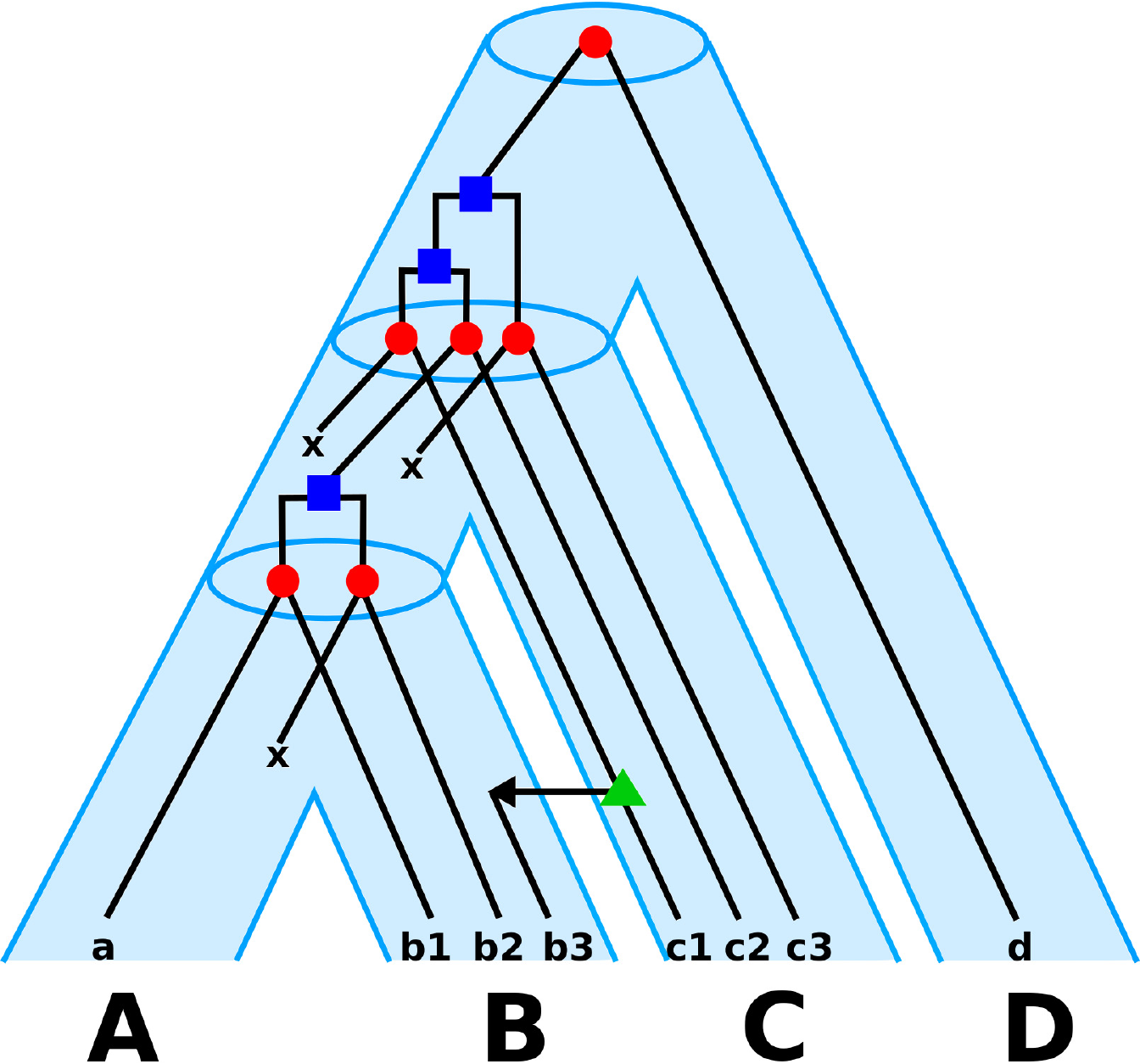}  
  \end{minipage}
  \caption{Example of an evolutionary scenario showing the ``true''
    evolution of a gene family evolving along the species tree (shown as
    blue tube-like tree). The corresponding gene tree $T$ appears embedded
    in the species tree $S$.  The speciation vertices in the gene tree (red
    circuits) appear on the vertices of the species tree (blue ovals),
    while the duplication vertices (blue squares) and the HGT-vertices
    (green triangles) are located on the edges of the species tree.  Gene
    losses are represented with ``{\textsf{x}}''.  The gene-tree $T$
    uniquely determines the relationships between the genes by means of the
    event at the least common ancestor $\lca_T(x,y)$ of distinct genes
    $x,y\in \Gen$. There is a clear distinction between orthologs
    (comprised in $R_o$ and indicated via red edges), paralogs (comprised
    in $R_p$ and indicated via non-drawn edges), as well as xenologs, that
    are neither orthologs nor paralogs (comprised in $R_x$ and indicated by
    green directed arcs).  }
  \label{fig:trueHist}
\end{figure}

The mathematical framework for orthology relations in terms of symbolic
ultrametrics \cite{BD98,Hellmuth:13a}, on the other hand, naturally
accommodates more than two types of events associated with the internal nodes of
the gene tree \cite{HW:16b}. It is appealing, therefore, to think of a HGT event as
different from both speciation and duplication, in line with \cite{Gray:83}
where the term ``xenologous'' was originally introduced, see Figure
\ref{fig:trueHist} for an illustrative example. The inherently asymmetric
nature of HGT events, with their unambiguous distinction between the
vertically transmitted ``original'' and horizontally transmitted ``copy''
furthermore suggests to relax the symmetry assumption and explore a
generalization to directed graphs and symbolic ``quasi-metrics''.  From the
mathematical point of view it seems natural to ask which systems of binary
relations on a set $V$ (of genes) can be represented by a (phylogenetic)
tree with leaf set $V$ and a suitable labeling (of event types) on the
internal nodes of $T$. To this end the theory of 2-structures
\cite{ER1:90,ER2:90} provides an interesting starting point.

This contribution is organized as follows.  In Section \ref{sec:prelim}, we
present the basic and relevant concepts used in this paper. In particular,
we briefly survey existing results concerning di-cographs, symbolic
ultrametrics and 2-structures.  In Section \ref{sec:charact}, we establish
the characterization of 2-structures that have a particular
tree-representation, so-called uniformly non-prime (\unp) 2-structures, in
terms of symbolic ultrametrics, di-cographs and so-called 1-clusters that
are obtained from the tree-representation of the respective
di-cographs. These results are summarized in Theorem \ref{thm:charactALL}.
In Section \ref{sec:alg}, we use the characterization of \unp\ 2-structure
to design a conceptual quite simple quadratic-time algorithm to recognize
whether a 2-structure is \unp, and in the positive case, to construct the
respective tree-representation (Theorem \ref{thm:algALL}). Furthermore, we
are concerned with editing problems to obtain a \unp\ 2-structure, showing
the NP-completeness of the underlying decision problems (Theorem
\ref{thm:npc}) and describe integer linear programming formulations to
solve them.

\section{Preliminaries}
\label{sec:prelim}

\subsection{Basic Notation} 

Throughout this contribution all sets are finite. We say that two sets $A$
and $B$ \emph{overlap}, in symbols $A \overlap B$, if $A\cap B\neq
\emptyset$ and neither $A\subseteq B$ nor $B\subseteq A$. Given a set $V$
we identify binary relations $R\subseteq V\times V$ with the \emph{directed
  graphs} (\emph{di-graphs} for short) $G=(V,R)$ with vertex set $V$ and
\emph{arc} set $R$. Throughout, we are concerned with irreflexive relations
or, equivalently, loop-free digraphs, i.e., $R\subseteq \Virr \coloneqq
V\times V\setminus \{(v,v)|v\in V\}$. For an arc $e=(x,y)\in \Virr$ we
write $e^{-1}$ to designate the reverse arc $(y,x)$. \emph{(Undirected)
  graphs} are modeled by edge sets $E\subseteq \binom{V}{2}$ taken from
the set of unordered pairs of vertices.

An undirected graph $G=(V,E)$ is \emph{connected} if for any two vertices
$x,y\in V$ there is a sequence of vertices $(x,v_1,\dots,v_n,y)$, called
\emph{walk}, s.t.\ $\{x,v_1\},\{v_n,y\}$ and $\{v_i,v_{i+1}\}$, $1\leq
i\leq n-1$ are contained in $E$. A di-graph $G=(V,E)$ is \emph{(weakly)
  connected} if the undirected graph $G_u=(V,E_u)$ with $E_u=\{\{x,y\}\mid
(x,y)\in E\}$ is connected.  We say that a sequence of vertices
$S=(x,v_1,\dots,v_n,y)$ is a \emph{walk} in the di-graph $G=(V,E)$, if $S$
is a walk in the underlying undirected graph $G_u$.  A graph $H=(W,F)$ is a
\emph{subgraph} of $G=(V,E)$ if $F\subseteq W\times W$ for di-graphs or
$F\subseteq \binom{W}{2}$ for undirected graphs, $W\subseteq V$ and
$F\subseteq E$.  We will write $H\subseteq G$, if $H$ is a subgraph of $G$.
The subgraph $H=(W,F)$ is an \emph{induced} di-graph if in addition
$(x,y)\in W\times W$ and $(x,y)\in E$ implies $(x,y)\in F$.  The
corresponding condition in the undirected case reads
$\{x,y\}\in\binom{W}{2}$ and $\{x,y\}\in E$ implies $\{x,y\}\in F$.  A
\emph{connected component} of a \mbox{(di-)graph} is a connected subgraph
that is maximal w.r.t.\ inclusion. A di-graph $G=(V,E)$ is \emph{complete}
if $E=\Virr$,  and it is \emph{arc-labeled} if there is a map $\vp\colon E
\to \Upsilon$ that assigns to each arc a label $i\in \Upsilon$.

A \emph{tree} is a connected undirected graph that does not contain cycles.
A \emph{rooted} tree $T=(V,E)$ is a tree with one distinguished vertex
$\rho\in V$ called \emph{root}. The leaf set $L\subseteq V$ comprises all
vertices that are distinct from the root and have degree $1$. All vertices
that are contained in $V^0\coloneqq V\setminus L$ are called \emph{inner}
vertices. The first inner vertex $\lca(x,y)$ that lies on both unique paths
from two vertices $x$, resp., $y$ to the root, is called \emph{lowest
common ancestor} of $x$ and $y$. We write $L(v)$ for the set of leaves in
the subtree below a fixed vertex $v$, i.e., $L(v)$ is the set of all leaves
for which $v$ is located on the unique path from $x\in L(v)$ to the root of
$T$. The \emph{children} of an inner vertex $v$ are its direct descendants,
i.e., vertices $w$ with $\{v,w\}\in E(T)$ s.t. that $w$ is further away
from the root than $v$. An \emph{ordered} tree is a rooted tree in which an
ordering is specified for the children of each vertex. Hence, ordered trees
particularly imply a linear order $\leq$ of the leaves in $L$, and we say
that $x$ is left from $y$ iff $x<y$.

Two rooted trees $T_1$ and $T_2$ on the same leaf set $L$ are said to be
\emph{isomorphic} if there is a bijection $\psi:V(T_1)\to V(T_2)$ that
induces a graph isomorphism from $T_1$ to $T_2$ which is the identity on
$L$ and maps the root of $T_1$ to the root of $T_2$.

It is well-known that there is a one-to-one correspondence between
(isomorphism classes of) rooted trees on $V$ and hierarchies on $V$. A
\emph{hierarchy on $V$} is a subset $\mc{C}\subseteq 2^V$ such that (i)
$V\in\mc{C}$, (ii) $\{x\}\in\mc{C}$ for all $x\in V$, and (iii) $p\cap q\in
\{p,q,\emptyset\}$ for all $p,q\in\mc{C}$.  Condition (iii) states that no
two members of $\mc{C}$ overlap. Members of $\mc{C}$ are called
\emph{clusters}. The number of clusters in a hierarchy is bounded
\cite{Hellmuth:15a} and there is a well-known bijection between hierarchies
and trees \cite{sem-ste-03a}:
\begin{theorem}
  Let $\mc{C}$ be a collection of non-empty subsets of $V$.  Then, there is
  a rooted tree $T=(W,E)$ on $V$ with $\mc{C} = \{L(v)\mid v\in W\}$ if and
  only if $\mc{C}$ is a hierarchy on $V$.
  Moreover, the number of clusters $|\mc{C}|$ in a hierarchy $\mc{C}$ on $V$
  is bounded by $2|V| - 1$.	  
\label{A:thm:hierarchy}
\end{theorem}

\subsection{Di-Cographs}

Di-cographs are a generalization of the better-known undirected
cographs. Cographs are obtained from single vertices by repeated
application of disjoint union (\emph{parallel} composition) and graph join
(in this context often referred to as \emph{series} composition)
\cite{Corneil:81,BLS:99}. In the case of di-cographs, the so-called
\emph{order} composition is added, which amounts to a directed variant of
the join operation. More precisely, let $G_1, \dots , G_k$ be a set of $k$
disjoint digraphs. The disjoint union of the $G_i$s is the digraph whose
connected components are precisely the $G_i$s.  The series composition of
the $G_i$s is the union of these $k$ graphs plus all possible arcs between
vertices of different $G_i$s.  The order composition of the $G_i$s is the
union of these $k$ graphs plus all possible arcs from $G_i$ towards $G_j$,
with $1\leq i < j \leq k$.

We note that restricted to posets, di-cographs coincide with the
series-parallel orders \cite{Valdes:82}. Di-cographs are characterized by
the collection of forbidden induced subgraphs shown Fig.~\ref{fig:forb}
\cite{CP-06}. An undirected graph is a cograph if and only if it does not
contain $P_4$ as an induced subgraph \cite{Corneil:81}.

We emphasize that the results stated here are direct consequences of the
results for 2-structure in Section \ref{sec:2s}. However, since di-cographs
will play a central role for the characterization of certain 2-structures,
we treat them here separately.

Given an arbitrary (di)graph $G=(V,E)$, a \emph{(graph-)module} $M$ is a
subset $M\subset V$ such that for any $x\in M$ and $z\in V\setminus M$ it
holds that $(x,z)\in E$ if and only if $(y,z)\in E$ for all $y\in M$ and
$(z,x)\in E$ if and only if $(z,y)\in E$ for all $y\in M$. The subfamily of
so-called strong modules, i.e., those that do not overlap other modules,
form a hierarchy and is called \emph{modular decomposition} of $G$.  For a
given graph $G$ we denote with $\Ms(G)$ the set of its strong
modules. Since $\Ms(G)$ forms a hierarchy, there is an equivalent (ordered,
rooted) tree, that is well known as the \emph{modular decomposition tree}
of $G$ \cite{Moehring:84}. The (unique) modular decomposition tree of a
di-cograph is known as its \emph{cotree}. Its leaves are identified with
the vertices of the di-cograph and the inner vertices are labeled by the
composition operations. Conversely, any ordered tree with internal vertices
labeled by the operations \emph{parallel}, \emph{series}, or \emph{order},
defines a unique di-cograph on its leaf-set.

Any inner vertex of the modular decomposition tree $T$ of $G$ corresponds
to a strong module $L(v)\in \Ms(G)$. Moreover, each child $u$ of $v$ in $T$
corresponds to a strong module $L(u)\subsetneq L(v)$ so that there is no
other module $M\in \Ms(G)$ with $L(u)\subsetneq M \subsetneq L(v)$
\cite{hellmuth2015techniques}. Therefore, we refer to the module $L(u)$ as
a \emph{the child} of the module $L(v)$ if $u$ is a child of $v$ in the
modular decomposition tree.

\begin{figure}[tbp]
  \begin{center}
    \includegraphics[width=0.95\textwidth]{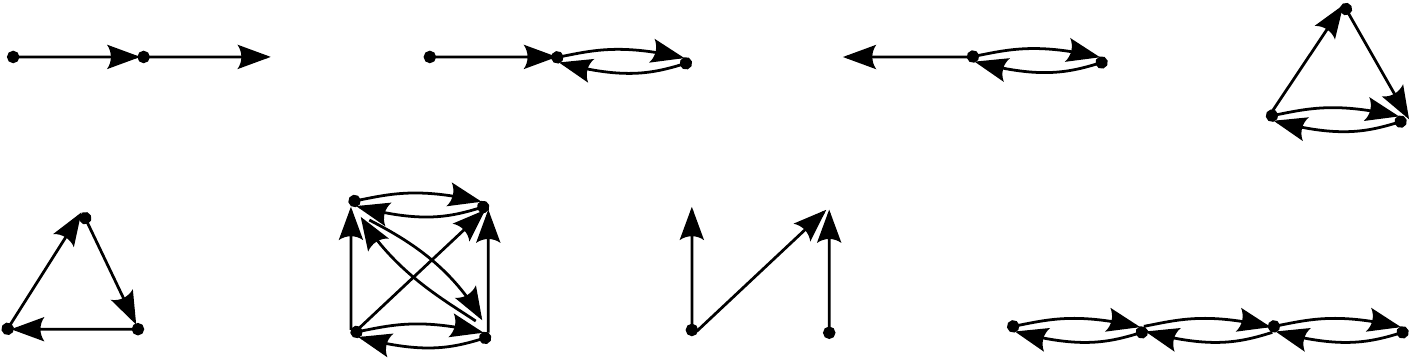}
  \end{center}
  \caption{Forbidden subgraphs for di-cographs. A di-graph $G$ is a
    di-cograph if and only if it does not contain one of these graphs as an
    induced subgraph. Following \cite{EHPR:96} we denote them from left to
    right $D_3, A, B, \overline{D_3}$ (in the 1st line), and $C_3,
    \overline{N}, N, P_4 $ (in the 2nd line).  A similar picture appeared
    in \cite{CP-06}.}
  \label{fig:forb}
\end{figure}

For simplicity, we use for a di-cograph $G$ and its respective cotree $T$
the labeling function $t:V^0(T)\to \{0,1,\overrightarrow{1}\}$ defined by
\begin{align*}
  t(\lca(x,y)) = \begin{cases}
    0,    & \text{ if } (x,y)(y,x)\notin E(G) \text{ (``parallel'')}   \\
    1,     & \text{ if } (x,y)(y,x)\in E(G)   \text{ (``series'')}  \\
    \overrightarrow{1}, &\text{ else } 
    \text{ (``order'')} 
 \end{cases}
\end{align*}
throughout this contribution. Since the vertices in the cotree $T$ are
ordered, the label $\overrightarrow{1}$ on some $\lca(x,y)$ of two distinct
leaves $x,y\in L$ means that there is an arc $(x,y) \in E(G)$, while $(y,x)
\notin E(G)$, whenever $x$ is placed to the left of $y$ in $T$.  For a
given cotree $T$ and inner vertex $v$, we will also call the strong module
$L(v)$ of a di-cograph parallel, series, or order, if it is labeled $0$,
$1$ and $\overrightarrow{1}$, respectively. The modular decomposition of a
digraph that is not a cograph also contains strong modules that are neither
parallel, nor series, nor order. Such modules are called \emph{prime}.

\subsection{Symbolic Ultrametrics}

Let $V$ and $\Upsilon$ be non-empty sets and let $\delta:\Virr \to
\Upsilon$, $(x,y)\mapsto \delta(xy)$ be a map that assigns to each pair
$(x,y)\in \Virr$ the unique label $\delta(xy) \in \Upsilon$.  Note that we
write $\delta(xy)$ instead of $\delta((x,y))$ to somewhat simplify the
notation. For two distinct vertices $x,y \in V$ we denote by
$D_{xy}\coloneqq\{\delta(xy),\delta(yx)\}$ the set of labels assigned to
the pairs $(x,y)$ and $(y,x)$. By construction $D_{xy}=D_{yx}$ and
$|D_{xy}|=1$ iff $\delta(xy) = \delta(yx)$. Given a map $\delta:\Virr \to
\Upsilon$ we define $G_i(\delta) = (V,E_i)$ with arc set $E_i = \{(x,y)\in
\Virr \mid \delta(xy) = i\}$ for all $i\in\Upsilon$. Finally, we define
$D_{xyz} \coloneqq \left\{ D_{xy}, D_{xz}, D_{yz} \right\}$ for any three
vertices $x,y,z \in V$. Note that $D_{xyz}$ is the set of distinct label
pairs assigned to the constituent unordered pairs of the 3-tuple, not the
set of distinct colors assigned to the six underlying ordered pairs.

\begin{definition} 
A map $\delta:\Virr\to\Upsilon$ is a \emph{symbolic ultrametric} on $V$ 
if it satisfies 
\begin{itemize}
\item[(U1)] $G_{i}(\delta)$ is a di-cograph for all $i\in \Upsilon$.
\item[(U2)] $|D_{xyz}|\leq 2$ for all $x,y,z \in V$.   
\end{itemize}
A symbolic ultrametric is called \emph{symmetric} if it satisfies 
\begin{itemize}
\item[(U3)] $\delta(xy) = \delta(yx)$ for all distinct $x,y \in V$.
\end{itemize}
\end{definition}
We will refer to axiom (U2) as the ``$\Delta(xyz)$-Condition'' or simply 
``Triangle-Condition''. 

If $\delta$ is symmetric, than (U1) and (U2) can be replaced by the
equivalent conditions:
\begin{itemize}
\item[(U1')] there exists no subset $\{x,y,u,v\}\in \binom{X}{4}$ 
  such that  $\delta(xy)=\delta(yu)=\delta(uv) \neq 
  \delta(yv)= \delta(xv)=\delta(xu)$.
\item[(U2')] $|\{\delta(xy), \delta(xz),\delta(yz)\}| \le 2$
  for all $x,y,z \in X$;\\
\end{itemize}
Condition (U1') identifies a pair of ``mono-chromatic paths'' $x-y-u-v$ and
$y-v-x-u$ as forbidden configuration. Equivalently, each of the graphs
$G_i(\delta)$ does not have an induced $P_4$, and thus it is an undirected
cograph.  Symmetric ultrametrics have been introduced in \cite{BD98} and
were studied subsequently, e.g., in \cite{Hellmuth:13a,HW:15,HW:16}.

An \emph{ultrametric} $d$ on $X$ is a real-valued symmetric map $d: X
\times X\to \mathbb{R}$ that (i) vanishes exactly on the diagonal and (ii)
satisfies $d(x,z) \leq \max\{d(x,y), d(y,z)\}$ for all $x,y,z\in
X$. Reading the real-valued distances as labels, we observe that $d$ is
also a (symmetric) symbolic ultrametric because, as shown in, e.g. 
\cite{sem-ste-03a},
the two larger distances coincide as a consequence of condition (ii). 

\subsection{2-Structures}
\label{sec:2s}

2-structures were introduced in \cite{ER1:90,ER2:90}. We refer to
\cite{Ehrenfeucht1995,ehrenfeucht1999theory,EHPR:96} for excellent
additional surveys. We follow the original terminology where possible. We
will, however, deviate at times to remain consistent with the literature on
co-graphs and symbolic ultrametrics.
\begin{definition}
  A (labeled) \emph{2-structure} is a triple $g = (V,\Upsilon,\vp)$ where
  $V$ and $\Upsilon$ are nonempty sets and  $\vp\colon \Virr \to \Upsilon$
  is a map.
\end{definition}
We refer to $V$ as the vertices and $\Upsilon$ as the labels. The
function $\vp$ maps each pair $(x,y)$, called an \emph{arc} of $g$ to a
unique label $\vp(xy)\coloneqq \vp((x,y))\in \Upsilon$. We will sometimes
write $V_g, \Upsilon_g$ and $\vp_g$ to emphasize that the vertex set,
label set, and the labeling function, resp., belong to the 2-structure $g$.

\emph{Isomorphic} 2-structures $g=(V,\Upsilon,\vp)$ and
$h=(V,\Upsilon',\vp')$, in symbols $g\simeq h$, differ only by a bijection
$\alpha:\Upsilon\to\Upsilon'$ of their labels, i.e.,
$\vp'(e)=\alpha(\vp(e))$ and $\vp(e)=\alpha^{-1}(\vp'(e))$ for all
$e\in\Virr$.

2-structures can be considered as arc-labeled complete graphs.
Conversely, every directed or undirected graph $G$ with vertex set $V$ has
a representation as a 2-structure by labeling the edges of the complete
graph by $0$ or $1$ depending on whether the arc is absent or present in
$G$. Thus we can interpret 2-structures as a natural generalizations of
(di-)graphs. Moreover, 2-structures are equivalent to a sets of disjoint
binary relations $R_1,\dots,R_k$ where each tuple $(x,y)$ has label $i$ iff
$(x,y)\in R_i $ or label $0$ if $(x,y)$ is not present in any of these
relations.

Extending the definition for symbolic ultrametrics above, we define for a
given 2-structure $g = (V,\Upsilon,\vp)$ and each $i\in \Upsilon$ the graph
$G_i(g) = (V,E_i)$ with arc set $E_i = \{(x,y)\in \Virr \mid \vp(xy) = i\}$.

Given a subset $X\subseteq V$ the \emph{substructure of} $g =
(V,\Upsilon,\vp)$ \emph{induced by} $X$ has vertex set $X$ and all arcs
$(a,b)\in\Xirr$ retain the color $\vp(ab)$, i.e., $g[X]\coloneqq
(X,\Upsilon, \vp' = \vp_{|\Xirr})$. A 2-structure $h$ is a substructure of
the 2-structure $g$ iff there is a subset $X\subseteq V_g$ so that
$h\simeq g[X]$.

\begin{definition}
  A \emph{module} (or clan) of a 2-structure is a subset $M\subseteq V$,
  such that $\vp(xz) = \vp(yz)$ and $\vp(zx) = \vp(zy)$ holds for all
  $x,y \in M$ and $z\in V\setminus M$.
\end{definition}
The empty set $\emptyset$, the complete vertex set $V_g$, and the
singletons $\{v\}$ are always modules. They are called the \emph{trivial}
modules of $g$. We will assume from here on, that a module is non-empty
unless otherwise indicated. The set of all modules of the 2-structure $g$
will be denoted by $\M(g)$.

\begin{lemma}
  A module of a 2-structure $g=(V,\Upsilon, \vp)$ is also a graph-module of
  $G_i(g)$ for all $i\in\Upsilon$.
\label{lem:module-graphmodule}
\end{lemma}
\begin{proof}
  Let $M\in \M(g)$ be an arbitrary module of $g$ and $i\in \Upsilon$ be
  some label. For $x\in M$ and $z\in V\setminus M$ we have $(x,z)\in
  E(G_i(g))$ if and only if $\vp(xz)=i$.  Since $M$ is a module in $g$, we
  then also have $\vp(yz)=i$ for all $y\in M$ and hence, 
  $(x,z) \in E(G_i(g))$ if and only if $(y,z)\in E(G_i(g))$ for all $y\in M$.
  Analogously, $(z,x)\in E(G_i(g))$ if and only if
  $(z,y)\in E(G_i(g))$ for all $y\in M$. Thus, $M$ is a module in $G_i(g)$.
  \qed
\end{proof}

The converse of Lemma \ref{lem:module-graphmodule} is not true in general,
Consider, for example, a 2-structure $g$ with $V_g=\{x,y,z\}$ and
$\vp_g(xy)=\vp_g(yx)=1$, $\vp_g(xz)=\vp_g(zx)=2$, and
$\vp_g(zy)=\vp_g(yz)=3$. One easily observes that $G_1(g)$ contains the
module $M=\{x,y\}$, since none of the edges $(x,z)$, $(z,x)$, $(y,z)$,
$(z,y)$ are contained in $G_1(g)$. However, since $\vp_g(xz)\neq \vp_g(yz)$
the set $M$ is not a module of $g$.
 
A very useful property of modules is summarized by
\begin{lemma}[\cite{ER1:90}, Lemma 4.11]
  Let $X, Y \in \M(g)$ be two disjoint modules of $g=(V,\Upsilon, \vp)$.
  Then there are labels $i,j\in \Upsilon$ such that $\vp(xy) = i$ and
  $\vp(yx) = j$ for all $x \in X$ and $y \in Y$.
  \label{lem:arcs-modules}
\end{lemma}

2-structures $g$ come in different types: 
\begin{enumerate} 
\item $g$ is \emph{prime} if $\M(g)$ consists of trivial modules only.
\item $g$ is \emph{complete} if for all $e,e'\in \Virr$, $\vp(e) = \vp(e')$
\item $g$ is  \emph{linear} if there are two distinct labels 
  $i, j \in \Upsilon$ such that the relations $<_i, <_j$ defined by
  \[ x <_i y \text{ iff } \vp(xy)=i, \text{ and } 
     x <_j y \text{ iff } \vp(xy)=j\]
  are linear orders of the vertex set $V_g$.
\end{enumerate} 
In particular, if $g$ is linear then there is a linear order $<$ of $V$
s.t.\ $x<y$ if and only if $\vp(xy)=i$ and $\vp(yx)=j$.  Clearly, if
$|V_g|=2$ all modules are trivial, and hence $g$ is prime. On the other
hand $|V_g|=2$ also implies that $g$ is either linear or complete.  For
$|V_g|\geq3$, however, the tree types of 2-structures are disjoint.

Not all 2-structures necessarily fall into one of these three types.  For
example, the 2-structure $g$ with $V_g=\{x,y,z\}$, $\vp_g(xy)=\vp_g(yx)=1$,
and $\vp_g(xz)=\vp_g(zx)=\vp_g(zy)=\vp_g(yz)=2$ is neither prime, nor
linear, nor complete.  We finally note that our notion of prime is called
``primitive'' in \cite{EHPR:96}.

A key concept for this contribution is
\begin{definition} 
A 2-structure is \emph{uniformly non-prime} (\unp) if it does not 
have a prime substructure $h$ of size $|V_h|\geq 3$.
\end{definition}

\begin{definition} 
  A module $M$ of $g$ is \emph{strong} if $M$ does not overlap with any
  other module of $g$.
\end{definition}
This notion was termed ``prime'' in \cite{EHPR:96}. We write
$\Ms(g)\subseteq \M(g)$ for the set of all strong modules of $g$.

While there may be exponentially many modules, the size of the set of
strong modules is $O(|V|)$ \cite{EHMS:94}.  For example, the 2-structure
$g=(V,\Upsilon,\vp)$ with $\vp(xy)=\vp(ab)$ for all $(x,y),(a,b)\in \Virr$
has $2^{|V|}$ modules, however, the $|V|+1$ strong modules are $V$ and the
singletons $\{v\}$, $v\in V$.

Since $V$ and the singletons $\{v\}$ are strong modules and strong modules
do not overlap by definition, we see immediately that $\Ms(g)$ forms a
hierarchy and by Thm.~\ref{A:thm:hierarchy} gives rise to a unique tree
representation $T_g$ of $g$, also called \emph{inclusion tree}. The
vertices of $T_g$ are (identified with) the elements of $\Ms(g)$.
Adjacency in $T_g$ is defined by the maximal proper inclusion relation,
that is, there is an edge $\{M,M'\}$ between $M,M'\in \Ms(g)$ iff
$M\subsetneq M'$ and there is no $M''\in \Ms(g)$ s.t.\ $M\subsetneq M''
\subsetneq M'$.  The root of $T_g$ is $V$ and the leaves are the singletons
$\{v\}$, $v\in V$.  Although $\Ms(g) \subseteq \M(g)$ does not represent
all modules, any module $M \in \M(G)$ is the union of children of the
strong modules in the tree $T_g$ \cite{Moehring:84,ER2:90}.  Thus, $T_g$
represents at least implicitly all modules of $g$.

The hierarchical structure of $\Ms(g)$ implies that there is a unique
partition $\Mmax(g) = \{M_1,\dots, M_k\}$ of $V_g$ into maximal (w.r.t.\
inclusion) strong modules $M_j\ne V_g$ of $g$ \cite{ER1:90,ER2:90}.  Since
$V_g\notin \Mmax(g)$ the set $\Mmax(g)$ consists of $k\ge 2$ strong
modules, whenever $|V_g|>1$.

In order to infer $g$ from $T_g$ we need to determine the color $\vp(xy)$
of all pairs of distinct leaves $x,y$ of $T_g$ and thus of $V_g$.  Hence,
we need to define a labeling function $t_g$ that assigns the ``missing
information'' to the inner vertex of $T_g$. To this end, we will need to
understand the quotient $g/\Mmax(g)$, i.e., the 2-structure $(\Mmax(g),
\Upsilon, \vp')$ with $\vp'(M_i,M_j)=\vp(xy)$ for some $x\in M_i$ and $y\in
M_j$. Thus a quotient $g/\Mmax(g)$ is obtained from $g$ by contracting each
module in $M \in \Mmax(g)$ into a single node, and then inheriting the edge
classes from $g$.  By Lemma \ref{lem:arcs-modules}, the quotient
$g/\Mmax(g)$ is well-defined.  Although 2-structures are not necessarily
prime, linear or complete, their quotients $g/\Mmax(g)$ are always of one
of these types.

\begin{lemma}[\cite{ER2:90,EHPR:96}]
  Let $g$ be a 2-structure. Then the quotient $g/\Mmax(g)$ is either
  linear, or complete, or prime.  If $g$ is \unp, then $g/\Mmax(g)$ is
  either linear, or complete.
\end{lemma}

We shall say that an inner vertex $v$ of $T_g$ (or, equivalently, the
module $L(v)$) is linear, complete, or prime if the quotient
$g[L(v)]/\Mmax(g[(L(v)])$ is linear, complete, or prime, respectively.
In order to recover $g$ from $T_g$ one defines a labeling function $\sigma$
that assigns the quotient of $g[L(v)]$ to each inner vertex $v$,
i.e., 
\[\sigma(v)= g[L(v)]/\Mmax(g[L(v)]).\] 
This type of labeled tree representation is called $\mathrm{shape}(g)$ or
(strong) module decomposition of $g$ \cite{ER1:90,ER2:90,EHPR:96}. If we
restrict ourselves to \unp\ structures, the strong modules are ``generic''
in the sense that they are completely determined by the cardinalities of
their domains and the ordering of the vertices in $T_g$.

We can therefore define a simplified labeling function $t_g$ for \unp\
structures $g$. If $M\in \Mmax(g)$ is a complete module and 
$M_1,\dots, M_l$ are the children of $M$, then there is an $i\in \Upsilon$
s.t.\ for all vertices $x\in M_r, y\in M_s$, $r\neq s$ 
we have $\vp(xy)=\vp(yx)=i$.
Therefore, we can set  $t_g(M) = (i,i)$, implying that for all vertices 
$x,y$ with $\lca(x,y)=M$ it holds that $\vp(xy)=\vp(yx)=i$.
If $M\in \Mmax(g)$ is a linear module, then 
we can assume that the children of $M$ are ordered $M_1,\dots, M_l$ s.t.
$\vp(xy)=i$ and $\vp(yx)=j$ for some $i,j\in \Upsilon$
 if and only if $x\in M_r, y\in M_s$ and $1\leq r< s\leq l$.
Therefore, we can set  $t_g(M) = (i,j)$, implying that for all vertices 
$x,y$ with $\lca(x,y)=M$ and $x$ is left of $y$ in $T_g$ it 
holds that $\vp(xy)=i$ and $\vp(yx)=j$.

\begin{definition}
  The \emph{tree-representation} $(T_g,t_g)$ of a \unp\ 2-Structure $g$ is
  an ordered inclusion tree $T_g$ of the hierarchy $\Ms(g)$ together with a
  labeling $t_g:\Virr \to \Upsilon$ s.t.\ for all vertices $x,y \in V_g$ of
  the \unp\ 2-structure $g$ it holds that
  \[t_g(\lca(x,y)) = (i,j), \] 
  where $i=j$ if and only if $\vp_g(xy)=\vp_g(yx)=i$; and 
  $x$ is to the left of $y$ in $T_g$ if and only if 
  $\vp_g(xy)=i$, $\vp_g(yx)=j$, and $i\neq j$.
\end{definition}
Figure \ref{fig:unp2} shows an illustrative example.

We call two tree-representations $(T,t)$ and $(T',t')$ of a 2-structure $g$
\emph{isomorphic} if $T$ and $T'$ are isomorphic via a map $\psi:V(T)\to
V(T')$ such that $t'(\psi(v))= t(v)$ holds for all $v\in V(T)$. In the
latter we write $(T,t)\simeq (T',t')$

\begin{lemma}[\cite{ER2:90,EHPR:96}]
  For any 2-structures $h, g$ we have $(T_g,t_g) \simeq (T_h,t_h) $ if and
  only if $h\simeq g$.
\end{lemma}

The tree representation $(T,t)$ of $g$ contains no prime nodes if and only
if $g$ has no induced substructures of small size that are prime
\cite[Thm.\ 3.6]{EHPR:96}:
\begin{theorem}
  If $g$ is a 2-structure then the following statements are equivalent:
\begin{enumerate}
\item $g$ is \unp.
\item The tree-representation $(T_g,t_g)$ of $g$ has no inner vertex $v$
  labeled prime, i.e., the quotient $g[L(v)]/\Mmax(g[(L(v)])$ is always
  linear or complete.
\item $g$ has no prime substructure of size $3$ or $4$.
\end{enumerate}
\label{thm:char-old}
\end{theorem}
In particular, if $g$ is \unp, then every substructure on a subset $X$
with $|X|=3$ or $|X|=4$ has at least one non-trivial module, i.e, a module
$M\subseteq X$ with $|M|\geq 2$.

\begin{figure}[tbp]
  \begin{center}
    \includegraphics[width=\textwidth]{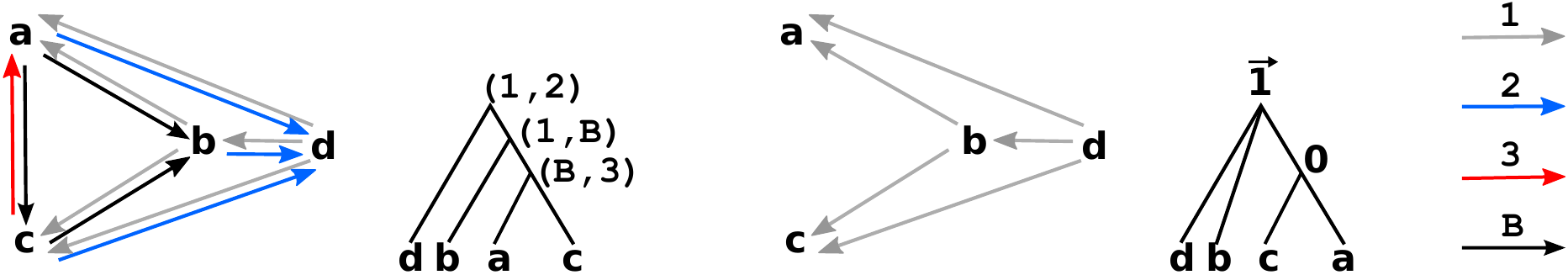}
  \end{center}
  \caption{Example of a \captionunp\ 2-structure $g$ with its
    tree-representation $(T_g,t_g)$ (1st and 2nd from left) and an
    underlying di-cograph $G_1(g)$ with respective cotree (3rd and 4th from
    left).  Colors are labeled with $\texttt{1,2,3,B}$ (right-most).  In
    fact, all underlying di-graphs $G_2(g), G_{3}(g)$ and $G_B(g)$ are
    di-cographs. To see that (U2) is satisfied for $\delta_g$ observe that
    $D_{abc}=\{\{B,1\}, \{B,3\}\}$, $D_{acd}=\{\{B,3\}, \{1,2\}\}$, and
    $D_{abd}=D_{bcd}=\{\{B,1\}, \{1,2\}\}$.}
  \label{fig:unp2}
\end{figure}

We next examine a particular subclass of 2-structures, the so-called
reversible 2-structures. As we shall see, they are simpler to handle than
general 2-structures. Nevertheless, there is no loss of generality as far
as modules are concerned.

\begin{definition}
  A 2-structure  $g=(V,\Upsilon, \vp)$ is \emph{reversible}, 
  if for all $e,f\in \Virr$, $\vp(e)=\vp(f)$ implies that  
  $\vp(e^{-1})=\vp(f^{-1})$. 
\label{def:rev-g}
\end{definition} 	
Equivalently, $g$ is reversible, if for each label $i \in \Upsilon$ there
is a unique label $j\in \Upsilon$ such that $\vp(x,y)= i$ implies
$\vp(y,x)= j$.

The definition of modules simplifies for reversible 2-structures. It
suffices to require that $M$ satisfies $\vp(xz) = \vp(yz)$ for 
all $z\in V\setminus M$ and $x,y\in M$, because $\vp(xz)=\vp(yz)$
and reversibility implies $\vp(zx)=\vp(zy)$.

\begin{definition}
The \emph{reversible refinement} of a 2-structure $g$ is a 
2-structure $\rev(g) = (V_g, \Upsilon_{\rev(g)}, \vp_{\rev(g)})$
where $\vp_{\rev(g)}(e)=h((\vp_{g}(e),\vp_{g}(e^{-1})))$ and 
$h: \Sigma \to \Upsilon_{\rev(g)}$ is an arbitrary 
bijection where $\Sigma = \{(\vp_{g}(e),\vp_{g}(e^{-1}))\mid e\in \Vgirr \}$ 
denotes the set of ordered pairs of colors on each arc $e$ and
its reverse $e^{-1}$.  
\end{definition}
The labels of $\rev(g)$ are most easily understood as pairs of labels of
$g$. Since the label sets are treated as sets without additional structure
in the context of 2-structures we allow an arbitrary relabeling. Since we
can identify isomorphic 2-structures, we can rephrase the definition in the
following form:
\begin{lemma} 
  A 2-structure $h$ is the reversible refinement of the 2-structure $g$ if
  and only if (i) $V_h=V_g$ and (ii) for all $e,f\in\Virr$ holds
  $\vp_{h}(e) = \vp_{h}(f)$ if and only if $\vp_{g}(e) = \vp_{g}(f)$ and
  $\vp_{g}(e^{-1}) = \vp_{g}(f^{-1})$.
\end{lemma} 

Finally, let us recall some well-established results concerning
(reversible) 2-structures.
\begin{theorem}[\cite{ER1:90}]
  For every 2-structure $g$ the following properties hold:
  \begin{enumerate}
  \item $\rev(g)$ is reversible, i.e., $\rev(\rev(g))=\rev(g)$.
  \item $g$ is reversible iff $g =\rev(g)$ .
  \item $\M(g) = \M(\rev(g))$. 
  \item A 2-structure $h$ is a substructure of $g$ iff 
    $\rev(h)$ is a substructure of $\rev(g)$. 
    \label{item:same-modules}
  \end{enumerate}
  \label{thm:points}
\end{theorem}

\begin{remark}
  For a given 2-structure $g=(V,\Upsilon,\vp)$ we define the map
  $\delta_g\colon\Virr\to\Upsilon$ as $\delta_g(xy) = \vp(xy)$ for all
  distinct $x,y\in V$.  Hence, $G_i(\delta_g) = G_i(g)$ for all $i\in
  \Upsilon$.  Conversely, each symbolic ultrametric
  $\delta\colon\Virr\to\Upsilon$ gives raise to a 2-structure
  $g=(V,\Upsilon, \delta)$.  To simplify the language, we will say that $g$
  satisfies Condition (U1) and (U2) whenever $\delta_g$ satisfies (U1) and
  (U2).
\end{remark}

\section{Characterization of \secunp\ 2-structures}
\label{sec:charact}

Let $\mc{R} = \{R_1,\dots,R_n\}$ be a set of disjoint relations.  Or goal
is to characterize the $\mc{R}$ that are obtained from a common tree $T$ and a
suitable labeling $t$ of the inner vertices of $T$. To this end we need to
understand on the one hand the relationships by symbolic ultrametrics and
event-labeled trees, and on the other hand the connection between symbolic
ultrametrics and 2-structures. Moreover, we will establish the connection
between 2-structures $g$ and certain modules in the modular decomposition
of the underlying graphs $G_i(g)$. The main results of this section are
summarized in Theorem \ref{thm:charactALL}.

\subsection{2-structures and Symbolic Ultrametrics}

We begin this section with the characterization of reversible 2-structures
by means of symbolic ultrametrics, which will be generalized to
arbitrary 2-structures at the end of this subsection.

\begin{remark}
  In what follows, we choose the notion $\Delta(xyz)$ as a shortcut for
  ``\emph{the Condition (U2) must be fulfilled for the set $D_{xyz}$}''.
  Moreover, for any forbidden subgraph $K$ that might occur in the graph
  $G_j(g) = G_j(\delta_g)$ of some 2-structure $g$, we use the symbols
  $K^j(abc)$ and $K^j(abcd)$, resp., to designate the fact that
  $G_j(\delta_g)$ contains the forbidden subgraph $K$ induced by the
  vertices $a,b,c$, resp., $a,b,c,d$ in $G_j(\delta_g)$.
\end{remark}

\begin{proposition}
For every reversible 2-structure $g=(V,\Upsilon, \vp)$ the following 
two statements are equivalent: 
\begin{enumerate}
\item[(1)] $g$ is \unp.
\item[(2)] $\delta_g$ is a symbolic ultrametric. 
\end{enumerate}
\label{prop:charact}
\end{proposition}
\begin{proof}
  In order to prove Proposition \ref{prop:charact}, we have to show that
  \unp\ 2-structures are characterized by the conditions
  \begin{enumerate}
  \item[(U1)] $G_{i}(\delta_g)$ is a di-cograph for all $i\in \Upsilon$  and 
  \item[(U2)] for all vertices $x,y,z \in V$ it holds 
    $|\left\{ D_{xy}, D_{xz}, D_{yz}  \right\}|\leq 2$. 
  \end{enumerate}

  We will frequently apply the following argument without explicitly
  stating it every time: By definition, if $g$ is reversible then 
  $\vp(e)=\vp(f)$ iff
  $\vp(e^{-1})=\vp(f^{-1})$. Hence, for reversible $g$,  
  $D_{ab}\neq D_{xy}$ implies that $\vp(ab)\neq \vp(xy), \vp(yx)$ and
  $\vp(ba)\neq \vp(xy), \vp(yx)$.

  \smallskip \emph{$\Rightarrow$:} Let $g=(V,\Upsilon, \vp)$ be a
  reversible \unp\ 2-structure. If $|V|<3$ then (U1) and (U2) are trivially
  satisfied. Thus we assume w.l.o.g.\ that $|V|\geq 3$. Furthermore,
  suppose there is a label $i\in\Upsilon$ such that $G_i(\delta_g)$ is not
  a di-cograph, i.e., $G_i(\delta_g)$ contains one of the forbidden
  subgraphs. Since $g$ is reversible, the forbidden subgraphs
  $A,B,\overline{D_3}$, and $\overline{N}$ cannot occur.

  Now let $h$ be a substructure of $g$ with $|V_h|=3$ containing $D_3$ or
  $C_3$, or $|V_h|=4$ containing $P_4$ or $N$, respectively. It is not hard
  to check that for each of these four graphs and any two distinct vertices
  $a,b \in V_h$ there is always a vertex $v\in V_h\setminus\{a,b\}$ so that
  $\vp(av)\neq \vp{(bv)}$.  Therefore, $\{a,b\}$ cannot form a module in
  $h$.  For $P_4$ and $N$ one checks that for any three distinct vertices
  $a,b,c \in V_h$ and $v\in V_h\setminus\{a,b,c\}$ we always have
  $\vp(av)\neq \vp{(bv)}$, or $\vp(av)\neq \vp{(cv)}$, or $\vp(bv)\neq
  \vp{(cv)}$, so that $\{a,b,c\}$ cannot form a module in $h$. Thus, $h$
  contains only trivial modules and, hence, is prime. This contradiction
  implies that (U1) must be fulfilled.

  Since $g=(V,\Upsilon,\vp)$ has a tree representation without prime nodes,
  and since three distinct leaves can have at most two distinct least
  common ancestors, Condition (U2) must hold as well.

  \smallskip\emph{$\Leftarrow$:} Now assume that $\delta_g$ is a symbolic
  ultrametric, i.e., condition (U1) and (U2) are fulfilled for a reversible
  2-structure $g$. In order to show that $g$ is \unp\ we have to
  demonstrate that all substructures $h$ of $g$ with $|V_h|=3$ and
  $|V_h|=4$ are non-prime (cf.\ Theorem \ref{thm:char-old}).

  \smallskip\noindent
  \emph{CASE: $h$ is a substructure of $g$ with $V_h=\{a,b,c\}$.}\\
  Since $\Delta(abc)$ we may assume that $D_{ab}=D_{ac}$, otherwise we
  simply relabel the vertices.  If $|D_{ab}|=1$, then $\{b,c\}$ forms a
  module in $h$.  Assume that $|D_{ab}|=2$. There are two cases, either
  $\vp(ab)=\vp(ac)$, then $\{b,c\}$ is a module in $h$, or
  $\vp(ba)=\vp(ac)=i$. In the latter case, $\vp(bc)=i$ since otherwise
  either or $D^i_3(abc)$ or $C_3^i(abc)$ would occur. Therefore, 
  $\{a,c\}$ forms a module in $h$. 

  Hence, in all cases, a substructure $h$ of $g$ with $V_h=\{a,b,c\}$ forms
  a non-prime structure. 

  \smallskip\noindent
  \emph{CASE: $h$ is a substructure of $g$ with $V_h=\{a,b,c,d\}$.}\\
  Since $\Delta(abc)$ we can assume that $D_{ab}=D_{ac}$, otherwise
  relabel the vertices. 
	
  Assume first that $|D_{ab}|=1$ and thus, 
  $\vp(ab)=\vp(ba)=\vp(ac)=\vp(ca)=i$ for some $i\in \Upsilon$.
  Since $\Delta(acd)$ we have the three distinct cases
  \begin{itemize}
  \item[(i)] $\vp(ad)=\vp(cd)=i$, 
  \item[(ii)] either (A) $\vp(ad)=i$ or (B) $\vp(cd)=i$
  \item[(iii)] neither $\vp(ad)=i$ nor $\vp(cd)=i$.
  \end{itemize}
  In Case (i) and (iiA), $\{b,c,d\}$ is a module in $h$. 
  In Case (iiB), the arc $(bc)$ or $(bd)$ must be labeled
  with $i$ as otherwise there is $P_4^i(abcd)$.
  If $\vp(bc)=i$, then $\{a,b,d\}$ is a module in $h$. 
  If $\vp(bd)=i$, then $\{a,d\}$ is a module in $h$. 

  Consider now Case (iii). Since $\Delta(acd)$, it follows
  that $D_{ad}=D_{cd}$ and in particular, $i\notin D_{ad}=D_{cd}$, 
  since $g$ is reversible.
  Let first $|D_{ad}|=|\{j\}|=1$. Since $\Delta(abd)$, we have that
  either $\vp(bd)=j$, in which case $\{a,b,c\}$ is a module in $h$	
  or $\vp(bd)=i$, which implies that $\vp(bc)=i$, since otherwise
  $P_4^i(abcd)$. In the latter case, $\{a,c,d\}$ forms a module in $h$.
  If  $|D_{ad}|=2$, we have only the case that  $\vp(ad)=\vp(cd)=j$
  for some $j\in \Upsilon$.
  In the two other cases $\vp(ad)=\vp(dc)=j$ or $\vp(da)=\vp(cd)=j$
  we would obtain $D_3^j(adc)$. 
  Since $\Delta(abd)$, we obtain that	either 
  (I) $\vp(bd)=i$, (II) $\vp(bd)=j$ or (III) $\vp(db)=j$. 
  Case (I) implies that $\vp(bc)=i$ as otherwise there is $P_4^i(abcd)$. 
  Hence, $\{a,c,d\}$ form a module in $h$. 
  In Case (II) $\{a,b,c\}$ is a module in $h$ and Case (III)
  cannot occur, as otherwise there is $D_3^j(abd)$.
  
  Now, assume  that $|D_{ab}|=2$. 
  Since $\Delta(abc)$, we can wlog.\ assume that 
  $D_{ab}=D_{ac}$, otherwise we relabel the vertices. 
  Hence, we have 
  either (I) $\vp(ab)=\vp(ac)=i$ or (II) $\vp(ba)=\vp(ac)=i$. 
  Note that in Case (I), $\vp(ba)=\vp(ca)=i'\neq i$ and in Case (II)
  $\vp(ab)=\vp(ca)=i'\neq i$. 
        
  Consider Case (I). 
  Since $\Delta(acd)$, we  have one of the four distinct cases
  \begin{itemize}
  \item[(i)] $D_{ac}=D_{ad}=D_{cd}$
  \item[(ii)] $D_{ac}=D_{ad}\neq D_{cd}$
  \item[(iii)] $D_{ac}=D_{cd}\neq D_{ad}$
  \item[(iv)] $D_{ac}\neq D_{cd}$ and $D_{ac}\neq D_{ad}$
  \end{itemize}
  
  In Case (Ii) it is not possible to have 
  $\vp(cd)=\vp(da)=i$ as otherwise there is 
  $C^i_{3}(acd)$. If $\vp(ad)=i$, then 
  $\{b,c,d\}$ is module in $h$. 
  If $\vp(da)=i$,	then $\vp(db)=\vp(dc)=i$,
  since otherwise there is $D_3^i(abd)$, $D_3^i(acd)$, 
  $C_3^i(abd)$ or $C_3^i(acd)$. In that case, $\{a,b,c\}$
  is a module in $h$. 	

  In Case (Iii) it is not possible to have
  $\vp(da)=i$, since otherwise there is $D_3^i(acd)$. 
  Thus, $\vp(ad)=i$ and therefore, 
  $\{b,c,d\}$ forms a module in $h$. 
  
  In Case (Iiii) it is not possible to have
  $\vp(cd)=i$, since otherwise there is $D_3^i(acd)$. 
  Hence, $\vp(dc)=i$. But then, at least one
  of the remaining arcs $(bc), (cb), (bd), (db)$
  must have label $i$, since otherwise there 
  is $N^i(abcd)$. 
  If $\vp(bc)=i$, then $\{a,b,d\}$ is a module in $h$.
  If $\vp(cb)=i$, then $\vp(db)=i$ as otherwise there
  is $D_3^i(bcd)$ or $C_3^i(bcd)$. Hence, $\{a,c,d\}$
  is a module in $h$. 
  The case $\vp(bd)=i$ is not possible, since then
  there is $D_3^i(abd)$.
  If $\vp(db)=i$, then $\{a,d\}$ is a module in $h$.
  
  In Case (Iiv) and since $\Delta(acd)$, we have 		
  $D_{ad}=D_{cd}$. 
  If $D_{ad}=\{j\}$ and thus, $|D_{ad}|=1$, then $\Delta(abd)$ implies that
  either $\vp(bd)=\vp(db)=j\neq i$, or
  $\vp(bd)=i$, or $\vp(db)=i$. 
  If $\vp(bd)=j\neq i$, then $\{a,b,c\}$ is a module
  in $h$. 
  The case $\vp(bd)=i$ cannot happen, since otherwise
  there is $D^i_3(abd)$. If $\vp(db)=i$, then either
  $\vp(bc)=i$ or $\vp(cb)=i$, otherwise there is 
  $N^i(abcd)$. The case $\vp(bc)=i$ is not possible, 
  otherwise there is $D_3^i(bcd)$. If  $\vp(cb)=i$, 
  then $\{a,c,d\}$ is a module in $h$. 
  
  Assume now that in Case (Iiv) we have $|D_{ad}|=2$. 
  Again, since $\Delta(acd)$, we have 		
  $D_{ad}=D_{cd}$. Assume that $j\in D_{ad}$.
  There are two case, either $\vp(ad)=\vp(cd)=j\neq i$
  or $\vp(ad)=\vp(dc)=j\neq i$. However, the latter
  case is not possible, otherwise there is $D_3^j(acd)$. 
  Hence, let $\vp(ad)=\vp(cd)=j\neq i$.
  Since $\Delta(abd)$ we can conclude that
  either $\vp(bd)=i$, or $\vp(db)=i$, or	
  $\vp(bd)=j$, or $\vp(db)=j$. 
  The cases  $\vp(bd)=i$ and $\vp(db)=j$ are not
  possible, otherwise there is
  $D^i_3(abd)$ and $D^j_3(abd)$, respectively. 
  If $\vp(db)=i$, then $\vp(bc)=i$ or $\vp(cb)=i$, otherwise there is 
  $N^i(abcd)$. This case can be treated as in the previous
  step and we obtain the module $\{a,c,d\}$ in $h$. 
  If $\vp(bd)=j$, then $\{a,b,c\}$ is a module in $h$. 
  
  Consider now Case (II) $\vp(ba)=\vp(ac)=i$, 
  and $\vp(ab)=\vp(ca)=i'\neq i$. 
  Hence, $\vp(bc)=i$, otherwise there is 
  $D_3^i(abc)$ or $C_3^i(abc)$.	
  Again, since $\Delta(acd)$, we have one of the four distinct cases
  (i), (ii), (iii) or (iv), as in Case (I). 
  
  Consider the Case (IIi). If $\vp(dc)=i$, then 
  $\{a,b,d\}$ is a module in $h$.
  Thus, assume $\vp(cd)=i$. The case $\vp(da)=i$ is
  not possible, since then there is $C_3^i(acd)$. 
  If $\vp(ad)=i$, then $\vp(bd)=i$, otherwise
  there is  $D_3^i(abd)$ or  $C_3^i(abd)$.
  Now, $\{a,c,d\}$ is a module in $h$.
  
  Now, Case (IIii). 
  The case $\vp(da)=i$ is not possible, otherwise there
  is $D_3^i(acd)$ and thus, $\vp(ad)=i$.  
  Then $\vp(bd)=i$, otherwise there is 
  $D^i_3(abd)$ or $C^i_3(abd)$. Therefore,
  $\{a,c,d\}$ is a module in $h$. 

  Consider the Case (IIiii).
  The case $\vp(cd)=i$ is not possible, otherwise there
  is $D_3^i(acd)$. Thus, $\vp(dc)=i$ and therefore, 
  $\{a,b,d\}$ is a module in $h$.
  
  In Case (IIiv) and since $\Delta(acd)$, we have 		
  $D_{ad}=D_{cd}$. 
  If $D_{ad}=\{j\}$ and thus, $|D_{ad}|=1$, then $\Delta(abd)$ implies
  that either $\vp(bd)=j\neq i$, or $\vp(bd)=i$, 
  or $\vp(db)=i$. 
  If $\vp(bd)=j\neq i$, then $\{a,b,c\}$ is a module
  in $h$. If $\vp(bd)=i$, then
  $\{a,c,d\}$ is a module in $h$.
  The case $\vp(db)=i$ cannot happen, otherwise there
  is $D_3^i(bcd)$.
  
  If $|D_{ad}|=2$ and $j\in D_{ad}$, then there are two cases
  either $\vp(ad)=\vp(cd)=j\neq i$
  or $\vp(ad)=\vp(dc)=j\neq i$.
  However, the latter
  case is not possible, otherwise there is $D_3^j(acd)$. 
  Hence, let 	$\vp(ad)=\vp(cd)=j\neq i$.
  Since $\Delta(abd)$ we can conclude that
  either $\vp(bd)=i$, or $\vp(db)=i$, or	
  $\vp(bd)=j$, or $\vp(db)=j$.
  The cases  $\vp(db)=i$ and $\vp(db)=j$ are not
  possible, otherwise there is
  $D^i_3(abd)$ and $D^j_3(abd)$, respectively. 
  If $\vp(bd)=i$ or $\vp(bd)=j$, then 
  $\{a,c,d\}$, resp., $\{a,b,c\}$ is a module in 
  $h$. 
  
  In summary, in each of the cases a substructure $h$ of $g$ with
  $3$ or $4$ vertices is non-prime whenever (U1) and (U2) holds. 
  Thus $g$ is \unp.
\qed
\end{proof}

The next step is to generalize Proposition \ref{prop:charact} to arbitrary
2-structures. To this end, we first prove two technical results:
\begin{lemma}
  Let $g=(V,\Upsilon,\vp)$ be a 2-structure.  Then condition (U2) is
  satisfied for $g$ if and only if (U2) is satisfied in $\rev(g)$.
  \label{lem:tri-rev}
\end{lemma}
\begin{proof}
  Condition (U2) is satisfied in $g$ if and only if 
  $|D_{abc}| \leq 2$ in $g$ for all $a,b,c\in V$ if and only if
  there are two arcs $e,f$ in this triangle induced by $a,b,c$
  s.t.\ $\vp(e)=\vp(f)$ and $\vp(e^{-1})=\vp(f^{-1})$ if and only if
  $\vp_{\rev(g)}(e)=\vp_{\rev(g)}(f)$ (and thus, by reversibility of $g$, 
  $\vp_{\rev(g)}(e^{-1})=\vp_{\rev(g)}(f^{-1})$) 
  if and only if  $|D_{abc}| \leq 2$ in $\rev(g)$ for all $a,b,c\in V$.
\qed
\end{proof}

\begin{lemma}
  Let $g=(V,\Upsilon,\vp)$ be a 2-structure satisfying (U2). 
  Then $G_i(g)$ is a di-cograph for all $i\in\Upsilon$ if and only if 
  $G_j(\rev(g))$ is a di-cograph for all $j\in\Upsilon_{\rev(g)}$.
  \label{lem:cogr-rev}
\end{lemma}
\begin{proof}
  \emph{$\Rightarrow$:}
  Assume that all $G_i(g)$ are di-cographs in $g$ and there is 
  $j\in \Upsilon_{\rev(g)}$ s.t.\ $G_j(\rev(g))$ is not a cograph. 
  Then $G_j(\rev(g))$ contains a forbidden subgraph. 
  Since $g$ is reversible, only the subgraphs $D_3, C_3, N$ and $P_4$ are 
  possible. 
  Moreover, we have
  $G_j(\rev(g)) \subseteq G_k(g)$ for some $k\in \Upsilon$, 
  by construction of $\rev(g)$ and since 
  $\vp_{\rev(g)}(e)=\vp_{\rev(g)}(f)$ implies that
  $\vp(e)=\vp(f)$.
  
  If $G_j(\rev(g))$ contains $D_3$ induced by the vertices $x,y,z$,
  we can wlog.\ assume that the vertices are labeled so that
  $\vp_{\rev(g)}(xy)=\vp_{\rev(g)}(yz)=j\neq\vp_{\rev(g)}(xz)$, 
  $\vp_{\rev(g)}(zy)=\vp_{\rev(g)}(yx)=k\neq\vp_{\rev(g)}(zx)$ and 
  $j\neq k$.
  By construction of $\rev(g)$ we obtain 
  $\vp(xy)=\vp(yz)=j'$, $\vp(zy)=\vp(yx)=k'$ for some distinct 
  $j',k'\in \Upsilon$.
  However, since $g$ does not contain forbidden subgraphs in $G_{j'}(g)$, 
  there must be an arc connecting $x$ and $z$ with color $j'$.
  The possibilities $\vp(zx)=j'\neq \vp(xz)$ and 
  $\vp(zx)=\vp(xz)=j'$ cannot occur, since then $G_{j'}(g)$
  would contain a $C_3$ or $\overline D_3$ as forbidden subgraph. 
  Hence, it must hold that $\vp(xz)=j'$.  Analogously, one shows that
  $\vp(zx)=k'$. By construction of $\rev(g)$, we obtain 
  $\vp_{\rev(g)}(xz)=j$, and  
  $\vp_{\rev(g)}(zx)=k$, a contradiction. 
  
  If $G_j(\rev(g))$ contains a $C_3$ induced by the vertices $x,y,z$,
  we can wlog.\ assume that the vertices are labeled so that
  $\vp_{\rev(g)}(xy)=\vp_{\rev(g)}(yz)=\vp_{\rev(g)}(zx)\neq 
  \vp_{\rev(g)}(yx)=\vp_{\rev(g)}(xz)=\vp_{\rev(g)}(zy)$. 
  Thus, $\vp(xy)=\vp(yz)=\vp(zx)=j'$ and  
  $\vp(yx)=\vp(xz)=\vp(zy)=k'$. We have $j'\neq k'$ as otherwise 
  $\vp_{\rev(g)}(xy)=\vp_{\rev(g)}(yx)$. Therefore, 
  $G_{j'}(g)$ contains the forbidden subgraph $C_3$, a contradiction. 

  If $G_j(\rev(g))$ contains a $P_4$ induced by the vertices $a,b,c,d$,
  we can wlog.\ assume that the vertices are labeled so that
  $\vp_{\rev(g)}(e)=\vp_{\rev(g)}(f)=j$ for all 
  $e,f\in E'=\{(a,b),(b,a),(b,c),(c,b),(c,d),(d,c)\}$.
  For all these arcs $e,f\in E'$ it additionally holds that 
  $\vp(e)=\vp(f)=j'$. Moreover, for all other arcs 
  $e\in \{a,b,c,d\}^{\times}_{\mathrm{irr}}\setminus E'$
  it is not possible that $\vp(e)=\vp(e^{-1})=j'$, 
  as otherwise, $\vp_{\rev(g)}(e)=\vp_{\rev(g)}(e^{-1})=j$
  and the $P_4$ would not be an induced subgraph of 
  $G_j(\rev(g))$. By the latter argument and 
  since $G_{j'}(g)$ does not
  contain an induced $P_4$ there must be at least one  
  arc $e\in \{a,b,c,d\}^{\times}_{\mathrm{irr}}\setminus E'$
  with $\vp(e)=j'$, but $\vp(e^{-1})\neq j'$. 
  Now full enumeration of all possibilities (which we leave to the reader) 
  to set one, two, or three of these arcs to the label $j'$
  yields one of the forbidden subgraphs
  $\overline{D_3}, A, B$ or $\overline{N}$ in $G_{j'}(g)$, a contradiction. 
  
  If $G_j(\rev(g))$ contains an $N$ induced by the vertices $a,b,c,d$,
  we can wlog.\ assume that the vertices are labeled so that
  $\vp_{\rev(g)}(ba)=\vp_{\rev(g)}(bc)=\vp_{\rev(g)}(dc)=
  j\neq\vp_{\rev(g)}(ab)=\vp_{\rev(g)}(cb)=\vp_{\rev(g)}(cd)=k$.
  Thus, $\vp(ba)=\vp(bc)=\vp(dc)=j'\neq\vp(ab)=\vp(cb)=\vp(cd)=k'$. 
  Since $G_{j'}(g)$ is cograph, there must be 
  an arc $e\in E'=\{(a,c), (c,a), (a,d), (d,a), (b,d), (d,b)\}$	
  with $\vp(e)=j'$. Moreover, for this arc $e$ it must hold that 
  $\vp(e^{-1})\neq k'$ as otherwise, 
  $\vp_{\rev(g)}(e)=j$. 	
  The graph $G_k(\rev(g))$ also contains an $N$ induced by the 
  vertices $a,b,c,d$.
  Hence, by analogous arguments there is an  $f \in E'$, $e\neq f$
  with $\vp(f)=k'$ with  $\vp(f^{-1})\neq j'$.
  Assume first that $e$ is $(a,c)$ or $(c,a)$ and thus, 
  $D_{ac}=\{j',j''\}$ where $j''=j'$ is allowed. 	 
  If 	$f$ is $(a,d)$ or $(d,a)$, then 
  $D_{ad}=\{k',k''\}$ where $k''=k'$ is allowed. 	 
  But then $D_{acd} = \{\{k',j'\},\{j',j''\},\{k',k''\}\}$
  with $j'\neq k'$ and thus $|	D_{acd} | =3$ violating
  Condition (U2) in $g$, a contradiction. 
  If 	$f$ is $(b,d)$ or $(d,b)$, then $D_{bd}=\{k',k''\}$ 
  where $k''=k'$ is allowed. 
  Thus, $\{j',j''\}, \{k',j'\}\in D_{acd}$ and 
  $\{k',k''\}, \{k',j'\}\in D_{abd}$. 
  The only way to satisfy $|D_{acd}|=2$			
  and $|D_{abd}|=2$ is achieved by $D_{ad}=\{k',j'\}$. 
  However,  the case $\vp(e)=j'$ and 
  $\vp(e^{-1})=k'$ with $e\in E'$ is not allowed. 
  All other cases, starting with
  $e\in E'\setminus\{(a,c), (c,a)\}$ can be treated
  analogously. 

  \smallskip	
  \emph{$\Leftarrow$:}
  Assume that there is an $i \in\Upsilon$ s.t.\ $G_i(g)$
  is not a di-cograph and for all $j\in \Upsilon_{\rev(g)}$
  the di-graph $G_j(\rev(g))$ is a di-cograph. 
  Hence, $G_i(g)$ contains a forbidden subgraph.
  
  If $G_i(g)$ contains a forbidden subgraph 
  $D_3, A, B$
  then there are arcs $(a,b), (b,c)$ 
  contained
  in these forbidden subgraphs with
  $\vp(ab)=\vp(bc)=i$ but $\vp(ac)\neq i$ and $\vp(ca)\neq i$. 
  Moreover, since  $G_j(\rev(g))$ does not contain these forbidden 
  subgraphs for any 
  $j\in \Upsilon_{\rev(g)}$, we also obtain that $\vp(ab)=\vp(bc)=i$ but
  $\vp(ba)\neq \vp(cb)$. 
  But this implies that
  $|D_{abc}|=3$ in $g$, a contradiction to (U2).
  
  If $G_i(g)$ contains a forbidden subgraph 
  $\overline{D_3}$ or $C_3$, then there are arcs 
  $(a,b), (b,c)$ 	contained
  in these forbidden subgraphs with
  $\vp(ab)=\vp(bc)=i$ and $\vp(ba)\neq i, \vp(cb)\neq i$. 
  If $\vp(ba) = \vp(cb)$ and the case $\overline{D_3}$
  is contained  $G_i(g)$, then  $G_j(\rev(g))$ contains
  the $D_3$ as forbidden subgraph. 
  If $\vp(ba) = \vp(cb)$ and the case $C_3$
  is contained  $G_i(g)$, then $G_j(\rev(g))$ contains
  the $D_3$ or $C_3$ as forbidden subgraph. 
  Hence, $\vp(ba) \neq \vp(cb)$.
  For the case $\overline{D_3}$, we observe that 
  $|D_{abc}|=3$ in $g$, a contradiction to (U2).
  For the case $C_3$, we can conclude by analogous arguments, 
  $\vp(ba) \neq \vp(ca)$ and $\vp(cb) \neq \vp(ca)$
  and again, $|D_{abc}|=3$ in $g$, a contradiction.
  
  Similarly, if $N$ is contained in  $G_i(g)$
  then there are arcs $(b,a), (b,c), (d,c)$ contained in $N$ with
  $\vp(ba)=\vp(bc)=\vp(dc)=i$ and $\vp(e) \neq i$ for all 
  $e \in \{(a,c), (c,a), (b,d), (d,b)\}$.
  Since $G_j(\rev(g))$ does not contain $N$ it holds that
  $\vp(ab) \neq \vp(cb)$ or $\vp(cb) \neq \vp(cd)$.
  If $\vp(ab) \neq \vp(cb)$ then $|D_{abc}|=3$,
  as $\vp(ac) \neq i$ and $\vp(ca) \neq i$,
  a contradiction to (U2).
  On the other hand, if $\vp(cb) \neq \vp(cd)$ then $|D_{bcd}|=3$,
  as $\vp(bd) \neq i$ and $\vp(db) \neq i$,
  again a contradiction to (U2).
  
  The $P_4$ on four vertices $a,b,c,d$ cannot be contained in any 
  $G_i(g)$, since for any two arcs
  $e,f\in E'=\{(a,b),(b,a),(b,c),(c,b),(c,d),(d,c)\}$ of this $P_4$
  it still holds $\vp_{\rev(g)}(e)=\vp_{\rev(g)}(f) = i'$
  and for any arc $e$ not in $E'$, $\vp_{\rev(g)}(e)\neq i'$. 
  Hence, if $G_i(g)$ contains a $P_4$, then  $G_{i'}(\rev(g))$
  contains a $P_4$ as forbidden subgraph, a contradiction. 
  
  If $G_i(g)$ contains the forbidden subgraph $\overline{N}$
  on four vertices $a,b,c,d$, then for the three arcs $e_1,e_2,e_3$
  with $\vp(e_j)=\vp(e_j^{-1})=i$, 
  it still holds, that 
  $\vp_{\rev(g)}(e_j)=\vp_{\rev(g)}(e_j^{-1})=i'$, $1\leq j\leq 3$. 
  However, for the other arcs $f_1,f_2,f_3$ with 
  $\vp(f_j)=i\neq \vp(f_j^{-1})$, we can infer that 
  $\vp_{\rev(g)}(f_j)\neq i'$ and $\vp_{\rev(g)}(f_j^{-1})\neq i'$.
  Thus, $G_{i'}(\rev(g))$ contains a $P_4$ on the three edges $e_1$, 
  $e_2$, $e_3$ as forbidden subgraph, a contradiction. 
  \qed
\end{proof}

It is now easy to establish our first main result:
\begin{theorem}
  The following two statements are equivalent for all 2-structures
  $g=(V,\Upsilon, \vp)$:
  \begin{enumerate}
  \item[(1)] $g$ is \unp.
  \item[(2)] $\delta_g$ is a symbolic ultrametric. 
  \end{enumerate}
\label{thm:charact}
\end{theorem}
\begin{proof}	 
  Lemma \ref{lem:tri-rev} and Lemma \ref{lem:cogr-rev} together imply that
  $\delta_{g}$ is a symbolic ultrametric if and only $\delta_{\rev(g)}$ is
  a symbolic ultrametric.  Proposition~\ref{prop:charact} implies that
  $\delta_{\rev(g)}$ is a symbolic ultrametric if and only if $\rev(g)$ is
  \unp. Now recall that \unp\ 2-structures are defined in terms of their
  modules and that $\M(g)=\M(\rev(g))$ (cf.\ Theorem
  \ref{thm:points}(\ref{item:same-modules})).  Therefore, $\rev(g)$ is a
  \unp\ 2-structure if and only if $g$ is a \unp\ 2-structure. The theorem
  follows immediately.  
  \qed
\end{proof}

\subsection{2-Structures and 1-Clusters}

Assume that $g=(V,\Upsilon,\vp)$ is a 2-structure with the property that
$G_{i}(g)$ is a di-cograph for all $i\in \Upsilon$. Each di-cograph
$G_{i}(g)$ is represented by a unique ordered tree $T_i$, called cotree
\cite{Corneil:81,CP-06}. Recall, in our notation the label of an inner
vertex in the cotree is always one of $0,1,\overrightarrow{1}$.  We say
that a leaf set $L(v)$ is a \emph{1-cluster of $T_i$} if $v$ has a label
distinct from $0$. The set $\mc{C}^1_i$ of 1-clusters of $T_i$ therefore is
a subset of the clusters that form the hierarchy equivalent to $T_i$.
Consider the set
\[\mc{C}^1(g)\coloneqq \cup_{i\in \Upsilon} \mc{C}^1_i \cup 
  \left\{\{v\} \mid v\in V_g\right\}
\] 
comprising the 1-clusters for each $T_i$ and the singletons.

\begin{remark}
  \label{rem:cluster-ser-ord}
  Any two disjoint (graph-)modules $M,M'$ of a di-graph $G$ are either
  adjacent or non-adjacent, i.e., for each vertex of $x\in M$ and each
  vertex of $y\in M'$ there is an arc $(x,y)$ or $(y,x)$ in $G$ or there is
  no edge between any vertex of $M$ and any vertex of $M'$
  \cite{Moehring:85,EHPR:96}.  Now, let $G=(V,E)$ be a di-cograph, $M \in
  \Ms(G)$ a strong module of $G$ and $M_1, \ldots, M_l$ the children of $M$
  in the respective cotree, i.e., the inclusion-maximal elements of
  $\Ms(G[M])$. By construction, module $M$ is a 1-cluster of $G$ if and
  only if $M$ is a series or order module.  Moreover, $M$ is a
  1-cluster if and only if for two vertices $x \in M_i$ and $y \in M_j$,
  with $i \neq j$ there exists at least one of the arcs $(x,y) \in E$ or
  $(y,x) \in E$.
\end{remark}

\begin{figure}[t]
  \begin{center}
    \includegraphics[width=\textwidth]{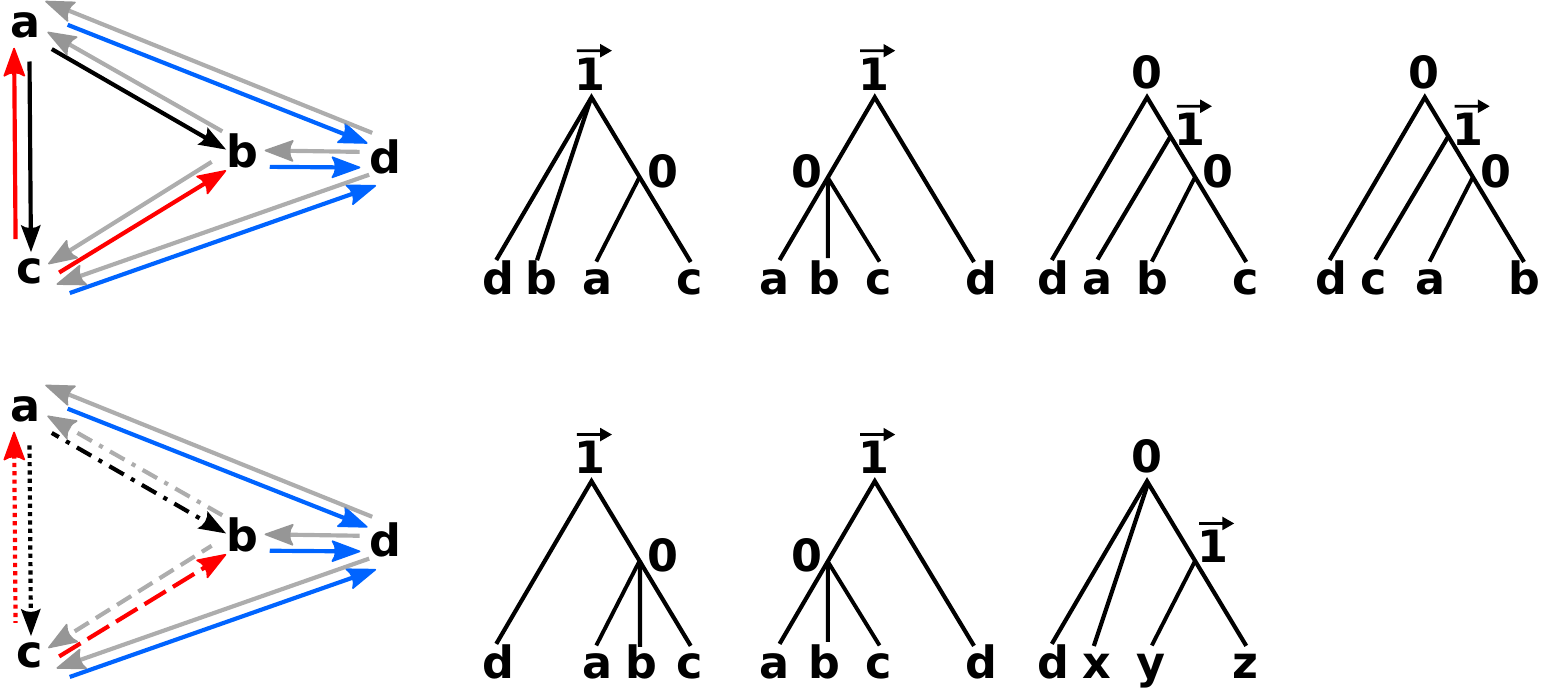}
  \end{center}
  \caption{Example of a non-reversible 2-structure $g$ (upper left part)
    and its reversible refinement $\rev(g)$ (lower left part). Labels are
    indicated by colored arcs.  The triangle-condition (U2) is violated in
    both $g$ and $\rev(g)$. Hence, neither 2-structure has a tree
    representation despite the fact that all $G_i(g)$ and $G_i(\rev(g))$
    are cographs. There are four different cographs, and hence cotrees, for
    $g$ and eight for $\rev(g)$. From these we obtain
    $\mc{C}^1(g)=\{\{a,b,c,d\},\{a,b,c\}, \{a\}, \{b\},
    \{c\},\{d\}\}$. Note that $\mc{C}^1(g)$ is a hierarchy even though $g$
    is not \captionunp.  Moreover, the clusters $\{a,b,c,d\}$, $\{a,b,c\}$ are
    both contained in two different cotrees. Thus the converse of
    Lemma~\ref{lem:unp-1clust_new} does not hold in general for
    non-reversible 2-structures.  The leaves $x,y,z$ of the right-most
    cotree of $\rev(g)$ can be chosen arbitrarily, as long as
    $\{x,y,z\}=\{a,b,c\}$.  Thus, $C^1(\rev(g))$ contains the cluster
    $\{a,b\}$ and $\{b,c\}$ and therefore does not form a hierarchy, cf.\
    Lemma \ref{lem:tri-clust} and Theorem \ref{thm:charactALL}.  
  }
\label{fig:t2}
\end{figure}

\begin{lemma} 
  Let $g=(V,\Upsilon,\vp)$ be a \emph{reversible} \unp\ 2-structure. Then
  $\mc{C}^1(g)$ is a hierarchy and, in particular, $\Ms(g)=\mc{C}^1(g)$.

  Moreover, for each cluster $M\in\mc{C}^1(g)$ there are at most two
  distinct di-cographs $G_i(g),G_j(g)$ s.t.\ $M\in \mc{C}^1_i$ and $M\in
  \mc{C}^1_j$. In other words, each cluster $M\in\mc{C}^1(g)$ appears as a
  1-cluster in at most two different cotrees.
  \label{lem:unp-1clust_new}
\end{lemma}
\begin{proof}
  We show that $\Ms(g)=\mc{C}^1(g)$. It then follows that $\mc{C}^1(g)$ is a
  hierarchy.

  Since g is \unp\ there is a tree representation $(T_g,t_g)$ of $g$. Let
  $v$ be an inner vertex in $T$ labeled with $(i,j)$.  By construction each
  vertex $v$ of this tree $T_g$ represents a strong module $L(v)$ of $g$.
  Hence, $L(v)\in \Ms(g)$ and thus, we can apply Lemma
  \ref{lem:module-graphmodule} and conclude that $L(v)$ is a module of
  $G_i(g)$.
 
  In fact, $L(v)$ is a strong module of $G_i(g)$. To see this, we first
  observe that all $a,b\in L(v)$ are contained in the same connected
  subgraph of both $G_i(g)$ and $G_j(g)$: Let $v_1,\dots, v_k$ be the
  children of $v$, ordered from left to right. If $i=j$, an ordering is not
  necessary. For leaves $x\in L(v_r)$ and $y\in L(v_s)$ we have now
  $\vp(xy)=i$ and $\vp(yx)=j$ if $r<s$. Hence, $x\in L(v_r)$ and $y\in
  L(v_s)$ and $r<s$ implies that $(x,y)\in E(G_i(g))$ and thus, all
  vertices $L(v)$ are contained in one connected subgraph of
  $G_i(g)$. Analogously, all vertices $L(v)$ are contained in one connected
  subgraph of $G_j(g)$.  Since $\vp(xy)=i$, $\vp(yx)=j$ if $r<s$ 
  and $g$ is reversible, we observe that $\vp(ab)=i$ if and only if
  $\vp(ba)=j$ for all $a,b\in V$.

  Now assume, for contradiction, that the module $L(v)$ is not strong in
  $G_i(g)$. Hence there is a further module $M$ in $G_i(g)$ s.t.\ $M
  \overlap L(v)$. Since $L(v)$ is a strong module in $g$, $M$ cannot be a
  module in $g$.  Since $M\overlap L(v)$, we have $M\cap L(v) \neq
  \emptyset$, $M \not\subseteq L(v)$ and $L(v) \not\subseteq M$. By the
  latter, and since all $a,b\in L(v)$ are contained in one connected
  subgraph of $G_i(g)$ there must be an arc $(u,x)$ or $(x,u)$ in $G_i(g)$,
  for some $x\in L(v)\setminus M$ and $u\in M \cap L(v)$. Wlog.\ assume
  that $(u,x)$ is an arc in $G_i(g)$, since the following arguments can be
  applied analogously for the case that $(x,u)$ is an arc in
  $G_i(g)$. Thus, we have $\vp(ux)=i$ and since $g$ is reversible,
  $\vp(xu)=j$.

  Since $M$ and $L(v)$ are modules in $G_i(g)$ and there is an arc $(u,x)$
  in $G_i(g)$, where particularly $x\in L(v)$ and $u\in M$, we can conclude
  that for all vertices $x'\in L(v)$, there is an arc $(y,x')$ in $G_i(g)$.
  This implies additionally that all $y\in M$ form an arc $(y,u)$, since
  $u$ is also contained in $L(v)$.  However, since $M$ is not a module in
  $g$ and $g$ is reversible, there must be a vertex $y\in M$ s.t.\ $\vp(yz)
  \neq \vp(uz)$ and $\vp(zy) \neq \vp(zu)$ for some $z\in V\setminus
  M$. Since $\vp(yx') = \vp(ux')$ for all $x'\in L(v)$, we can conclude
  that for the latter chosen vertex $z$ we have $z\in V\setminus (M \cup
  L(v))$.

  Since $\vp(yz) \neq \vp(uz)$, $\vp(zy) \neq \vp(zu)$, $g$ is reversible,
  and $M$ is a module in $G_i(g)$ with $u,y \in M$, we conclude that
  $\vp(yz)$, $\vp(zy)$, $\vp(uz)$, and $\vp(zu)$ must all be different from
  $i$. To see this, assume for contradiction that for some $e \in
  \{(y,z),(u,z)\}$ it holds that $\vp(e)=i$.  Since $M$ is a module in
  $G_i(g)$ it follows that $\vp(f)=i$ for $f \in \{(y,z),(u,z)\} \setminus
  \{e\}$, a contradiction to $\vp(yz) \neq \vp(uz)$. The same argument
  applies for $e,f \in \{(z,y),(z,u)\}$.
  
  Thus, assume that $\vp(yz) = l
  \neq \vp(uz) = k$ for some $l,k\in \Upsilon$ distinct from $i$. Since all
  $\vp(yz)$, $\vp(zy)$, $\vp(uz)$, $\vp(zu)$ are distinct from $i$ while
  $\vp(yu)=i$, and since $\Delta(yuz)$ we obtain $\vp(zy)=k$, and
  $\vp(zu)=l$.  Since $g$ is reversible, neither $\vp(uy)=k$, nor
  $\vp(uy)=l$.  But then there is $D^k_3(yuz)$ and $D^l_3(yuz)$ in $G_k(g)$
  and $G_l(g)$, a contradiction to (U1).

  Thus, $L(v)$ is a strong module in $G_i(g)$. All strong modules of
  $G_i(g)$ are represented in the respective cotree $T_i$. As already
  observed, since $t_g(v)=(i,j)$ for all leaves $x\in L(v_r)$ and $y\in
  L(v_s)$ we have $\vp(xy)=i$ if $r<s$. Hence, $(x,y)$ is an arc in
  $G_i(g)$ for all $x\in L(v_r), y\in L(v_s), r<s$.  If $i=j$ then even
  $(y,x)$ is an arc in $G_i(g)$.  Hence, $L(v)$ cannot be labeled with
  ``$0$'' in the cotree, because otherwise it is not possible to have all
  arcs $(x,y)$ with $x \in L(v_1)$ and $y\in L(v_i)$, $1<i\leq k$.  In
  particular, if $i\neq j$, then $L(v)$ must labeled ``$\ora{1}$'' in
  $G_i(g)$; if $i=j$, then the strong module $L(v)$ must labeled ``1'' in
  $G_i(g)$.  Hence, $L(v)$ is contained in $\mc{C}^1_i\subseteq
  \mc{C}^1(g)$. Therefore, $\Ms(g)\subseteq \mc{C}^1(g)$.

  \smallskip
  Now, let $L(v)\in \mc{C}^1_i\subseteq \mc{C}^1(g)$ be a strong module
  with label different from ``$0$'' obtained from the cotree $T_i$.
  Clearly, $|L(v)|>1$, since the singletons $\{v\}$ are by definition not
  contained in $\mc{C}^1_i$, albeit they are by construction contained in $
  \mc{C}^1(g)$.  Assume for contradiction that $L(v)$ is not a module in
  $g$. Since $g$ is reversible, there must be two vertices $a,b\in L(v)$ and
  $c\in V\setminus L(v)$ s.t.\ $\vp(ac)=j\neq \vp(bc)=k$.  In particular,
  $j$ and $k$ must both be distinct from $i$, as otherwise $L(v)$ would not
  be a module in $G_i(g)$.  Since $g$ is reversible, $\vp(ca)\neq \vp(cb)$
  and by analogous arguments as before, neither $\vp(ca)=i$ nor
  $\vp(cb)=i$.  The latter arguments and reversibility of $g$ imply that
  $\vp(xc)\neq i$ and $\vp(cx)\neq i$ for all $x\in L(v)$, as otherwise
  $L(v)$ would not be a module in $G_i(g)$ or $L(v)$ would be a module in 
  $g$, a contradiction. Since $L(v)\in \mc{C}^1_i$ the
  vertex $v$ has either label $1$ or $\ora 1$ in $T_i$. Thus, the subgraph
  in $G_i(g)$ induced by the vertices in $L(v)$ is connected.  Let
  $(a,v_1,\dots,v_n,b)$ be a walk in $G_i(g)$ with $v_i\in L(v)$ for $1\leq
  i \leq n$, that connects the vertices $a$ and $b$.  Since none of the
  labels $\vp(ac),\vp(ca), \vp(v_1c), \vp(cv_1)$ is $i$, and since
  $\Delta(av_1c)$ we can conclude that the label $\vp(ac)=j$ must occur
  on at least one of the arcs $(v_1c)$ or $(cv_1)$.  We distinguish the
  two cases (i) $\vp(v_1c)=j \neq \vp(cv_1)$ and (ii) 
  $\vp(v_1c)=\vp(cv_1)=j$.
  
  \emph{Case (i).} 
  Since none of the labels $\vp(v_2c),\vp(c v_2)$ is $i$, and $\Delta(v_1
  v_2c)$ we obtain that $\vp(v_2c) = j$ or $\vp(cv_2) = j$.  Repeating
  the latter, we obtain $\vp(v_nc) = j$ or $\vp(cv_n)=j$.  Since
  $\Delta(v_nbc)$ and none of the labels of $(v_nc)$ and $(cv_n)$ is
  $i$, but $\vp(bc)=k$, we can conclude that $\vp(cb)=j$ and the labels $j$
  and $k$ must occur on the two arcs $(v_n,c)$ and $(c, v_n)$.  The case
  $\vp(cv_n)=k$ cannot occur, since then there is $D_3^k(bcv_n)$.
  Thus, $\vp(v_nc) = k$ and by reversibility of $g$, $\vp(cv_n) =j$.  By
  analogous arguments, $\vp(cv_{n-1}) =j$ and $\vp(v_{n-1}c) = k$, and,
  iterative, $\vp(cv_{1}) =j$, $\vp(v_{1}c) = k$ and $\vp(ca)=k$.
  Since $g$ is reversible and $\vp(v_1a)=i$ or $\vp(av_1)=i$ we can
  conclude that there are forbidden subgraphs $D_3^j(acv_1)$ and
  $D_3^k(acv_1)$, a contradiction.

  \emph{Case (ii).}
  If $\vp(cv_1)=\vp(v_1c)=j$, then by analogous
  arguments as in Case (i),
  $\vp(cv_n)=\vp(v_nc)=j$. 
  But then $|D_{bcv_n}|=3$, violating (U2)
  and thus, $g$ is not \unp.
  Hence, $L(v)$ is a module of $g$. 
  
  It remains to show that $L(v)\in \mc{C}^1$ is also strong in $g$. Assume
  for contradiction that $L(v)$ is not strong in $g$ and hence, there is a
  module $M$ in $g$ with $|M|>1$ and $L(v)\overlap M$.  Lemma
  \ref{lem:module-graphmodule} implies that $M$ is also a module in
  $G_i(g)$ and hence, $M$ overlaps $L(v)$ in $G_i(g)$, a contradiction,
  since $L(v)$ is strong in $G_i(g)$.  Therefore, $\mc{C}^1(g)\subseteq
  \Ms(g)$.

  Taken together the latter arguments we can conclude that $\mc{C}^1(g) =
  \Ms(g)$ and thus, $\mc{C}^1(g)$ is a hierarchy.

  \smallskip 
  Finally, we show that each cluster $M\in \mc{C}^1(g)$ appears at most in
  two different cotrees.  Let $M\in \mc{C}^1(g)=\Ms(g)$ and assume that
  $t_g(M)=(i,j)$.  Analogously as in the first part of this proof, we can
  conclude that $M$ is a strong module in $G_i(g)$ with label $\ora{1}$, if
  $i\neq j$, and label $1$, if $i=j$.  Thus, $M\in \mc{C}^1_i$ in all cases
  and additionally, $M\in \mc{C}^1_j$, if $i\neq j$.
	
  It remains to show that there is no further $k\in \Upsilon$, $k\neq i,j$
  with $M\in \mc{C}^1_k$. Assume that this is not the case, and thus there is a
  di-cograph $G_k(g)$ s.t.\ $M$ is labeled $1$ or $\ora{1}$, in the
  respective cotree $T_k$.  Let $M_1,\dots,M_r$ be the children of $M$ in
  $T_g$ and $N_1,\dots,N_s$ be the children of $M$ in $T_k$.  Let $x_l$ be
  a vertex contained in $M_l$ for $1\leq l\leq r$.  Since every arc between
  distinct $x_l, x_{l'}$ is labeled $i$ or $j$, and in particular, not $k$,
  and since $x_1,\dots,x_r$ are contained in $M$, we can conclude that
  $x_1,\dots,x_r\in N_m$ for some $m\in \{1,\dots,s\}$.  Otherwise, there
  would be a label $k$ on some arc between some $x_l, x_{l'}$.  Now take a
  further vertex $y\in M_l$.  By analogous arguments for the vertices
  $x_1,\dots,x_{l-1},y,x_{l+1},\dots,x_r$ and since $x_1,\dots,x_r\in N_m$,
  we obtain that $y\in N_m$.  By induction, all vertices in $\cup_{i=1}^r
  M_i$ must be contained in $N_m$. Thus, $M=\cup_{i=1}^r M_i \subseteq N_m
  \subsetneq M$, a contradiction.  
  Therefore, $M\in \mc{C}^1_i, M\in \mc{C}^1_j$ but
  $M\notin \mc{C}^1_k$ for any $k\neq i,j$, from what the statement follows.
  \qed
\end{proof}

\begin{figure}[tbp]
  \begin{center}
    \includegraphics[width=\textwidth]{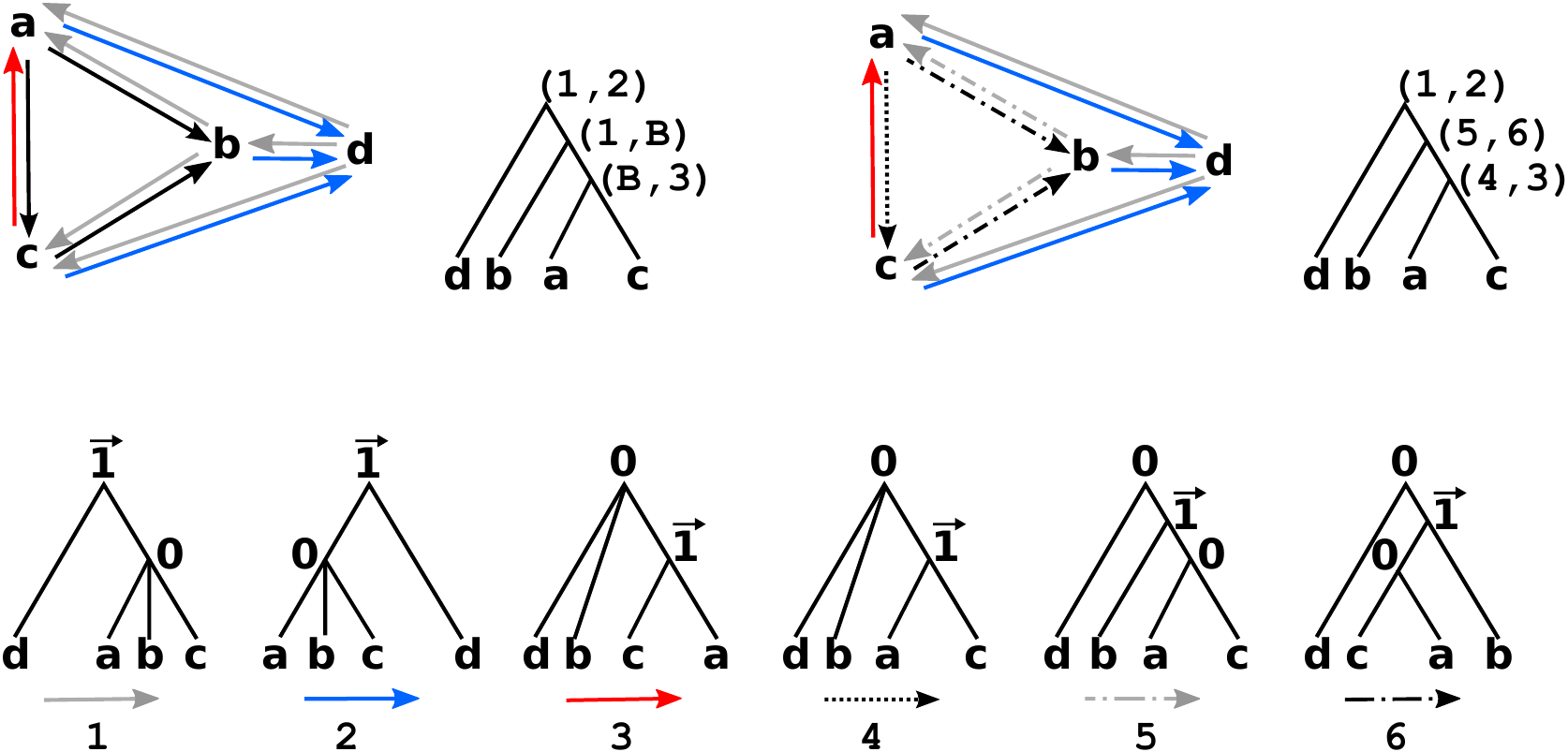}
  \end{center}
  \caption{Example of a non-reversible \captionunp\ 2-structure $g$ 
    with tree-representation $(T_g, t_g)$ (upper left)
    and its reversible refinement $\rev(g)$ with 
    $(T_{\rev(g)}, t_{\rev(g)})$ (upper right). Labels are
    indicated by colors and line-styles.  The label $B$ in
    the right-most tree corresponds to the solid black arc in $g$.  All
    $G_i(g)$ and $G_i(\rev(g))$ are di-cographs and the triangle condition
    is satisfied.  The cotrees corresponding to the $G_i(\rev(g))$ are
    shown in the second row. They generate $\mc{C}^1(\rev(g)) =
    \{\{a,b,c,d\},\{a,b,c\},\{a,c\}, \{a\}, \{b\}, \{c\}, \{d\}\}$, i.e., a
    hierarchy. By Lemma \ref{lem:tri-clust} and Theorem
    \ref{thm:charactALL}, $\mc{C}^1(\rev(g))=\Ms(\rev(g))= \Ms(g)$. The
    trees $T_g$ and $T_{\rev(g)}$ are isomorphic and differ only in
    the labels.}
  \label{fig:t1}
\end{figure}

Although it might be possible to derive a result similar to Lemma
\ref{lem:unp-1clust_new} for non-reversible 2-structures $g$ (with some
more elaborated technical arguments), the fact that $\mc{C}^1(g)$ is a
hierarchy and that each 1-cluster appears in at most 2 cotrees is not
sufficient to conclude that $g$ is \unp. Figure \ref{fig:t2} gives a
counterexample. Surprisingly, however, the triangle-condition (U2) and the
property that $\mc{C}^1(g)$ is a hierarchy, are equivalent for reversible
2-structures that fulfill (U1).

\begin{lemma} 
  Let $g=(V,\Upsilon,\vp)$ be a \emph{reversible} 2-structure s.t.\
  $G_{i}(g)$ is a di-cograph for all $i\in \Upsilon$.  Then the following
  statements are equivalent
  \begin{enumerate}		
  \item The Triangle-Condition $(U2)$ is satisfied for $g$
  \item $\mc{C}^1(g)$ is a hierarchy. In particular, $\Ms(g) =\mc{C}^1(g)$.
  \end{enumerate}
  \label{lem:tri-clust}
\end{lemma}
\begin{proof}
  If $g$ satisfies (U1) and (U2), then, by Theorem \ref{thm:charact}, $g$
  is \unp.  Now apply Lemma \ref{lem:unp-1clust_new}.

  Suppose $g$ does not satisfy (U2), i.e., there are three vertices
  $a,b,c\in V$ s.t.\ $\vp(ab)=i$ and $|D_{abc}|=3$. 
  Since $g$ is reversible we conclude
  that $\vp(ac)$, $\vp(ca)$, $\vp(bc)$, and $\vp(cb)$ are all distinct from
  $i$.  In the cotree $T_i$ of $G_i(g)$, $\lca(ab)$ must
  be labeled either $1$ or $\ora{1}$. Next we observe that $c$ cannot be
  descendant of $\lca(ab)$ in $T_i$, since otherwise we would have
  $\lca(ac)\in \{1,\ora{1}\}$ or $\lca(bc)\in \{1,\ora{1}\}$. This implies
  that at least one of the arcs $(ac)$, $(ca)$, $(bc)$, $(cb)$ must be present
  in $G_i(g)$, which is only possible iff $\vp$ mapped one of those arcs to
  the label $i$, a contradiction.  Therefore, there is a cluster in $T_i$
  that contains $a$ and $b$ but not $c$, and this cluster is also contained
  in $\mc{C}^1(g)$.  Now, let $\vp(ac)=j\neq i$.  Since $g$ is
  reversible, we know that $\vp(ab)$, $\vp(ba)$, $\vp(bc)$, and $\vp(cb)$
  are all distinct from $j$. Using the same argument as above, one can show
  that in the cotree $T_j$ there is a cluster containing $a$ and $c$ but
  not $b$, which is contained in $\mc{C}^1(g)$.  But then these two
  particular clusters overlap, and hence, $\mc{C}^1(g)$ is not a
  hierarchy. Since $\Ms(g)$ is a hierarchy, we can conclude that 
  $\mc{C}^1(g)\neq \Ms(g)$.  
  \qed
\end{proof}

\begin{corollary}
  Let $g=(V,\Upsilon,\vp)$ be a 2-structure s.t.\ $G_{i}(g)$ is a
  di-cograph for all $i\in \Upsilon$. The following statements are
  equivalent
  \begin{enumerate}		
  \item $g$ satisfies the Triangle-Condition (U2).
  \item $\mc{C}^1(\rev(g))$ is a hierarchy. 
  \item $\mc{C}^1(\rev(g)) = \Ms(g)$.
  \end{enumerate}
	\label{cor:g-rev1}
\end{corollary}
\begin{proof}
  By Lemma \ref{lem:tri-rev}, $g$ satisfies (U2) if and only if $\rev(g)$
  satisfies (U2). Together with Lemma \ref{lem:cogr-rev} this implies that
  $g$ satisfies (U1) if and only if $\rev(g)$ satisfies (U1). Therefore
  $\rev(g)$ satisfies (U1) and (U2), which, by Lemma \ref{lem:tri-clust},
  is equivalent to $\rev(g)$ satisfying (U1) and $\mc{C}^1(\rev(g))$
  being a hierarchy. In this case we have in particular 
  $\mc{C}^1(\rev(g)) = \Ms(\rev(g))$.  By Theorem
  \ref{thm:points}(\ref{item:same-modules}), $\M(\rev(g))=\M(g)$ and hence,
  $\Ms(\rev(g)) = \Ms(g)$.  Therefore, the statement is true.  
  \qed
\end{proof}

Collecting the results derived above we obtain the main result of 
this contribution:
\begin{theorem}[\unp\ 2-Structure Characterization]\ \\
  The following statements are equivalent for every 2-structure
  $g=(V,\Upsilon, \vp)$:
  \begin{enumerate}
  \item $g$ is \unp. \smallskip
  \item $\delta_g$ is a symbolic ultrametric. \smallskip 
  \item $g$ fulfills the following two properties:
    \begin{enumerate}	
    \item $G_{i}(g)$ is a di-cograph for all $i\in \Upsilon$. 
    \item $\mc{C}^1(\rev(g))$ is a hierarchy. 
      In particular, $\mc{C}^1(\rev(g)) = \Ms(g)$.
    \end{enumerate} \smallskip
  \item $\rev(g)$ fulfills the following two properties:
    \label{item:rev-g-unp}
    \begin{enumerate}	
    \item $G_{i}(rev(g))$ is a di-cograph for all $i\in \Upsilon_{\rev(g)}$. 
    \item $\mc{C}^1(\rev(g))$ is a hierarchy. 
    \end{enumerate}\smallskip
  \item $\rev(g)$ is \unp.
  \end{enumerate}
\label{thm:charactALL}
\end{theorem}
\begin{proof}
  The equivalence of Item (1.) and (2.) are already given in Theorem
  \ref{thm:charact}. Hence, $g$ satisfies (U1) and (U2).  By Corollary
  \ref{cor:g-rev1}, $\delta_g$ is a symbolic ultrametric if and only if
  (U1) and the condition that $\mc{C}^1(\rev(g))$ is a hierarchy (with
  $\mc{C}^1(\rev(g)) = \Ms(g)$) is satisfied. Thus, Item (2.) and (3.)
  are equivalent. By Lemma \ref{lem:cogr-rev} and since $g$ satisfies (U2),
  $g$ satisfies (U1) if and only if $\rev(g)$ satisfies (U1), and thus Item
  (3.) and (4.) are equivalent. Lemma \ref{lem:tri-clust} implies that the
  statement that $\rev(g)$ satisfies (U1) so that in addition 
  $\mc{C}^1(\rev(g))$ is a hierarchy is equivalent to the property that 
  $\rev(g)$ satisfies (U1) and (U2) and thus, $\delta_{\rev(g)}$ is a 
  symbolic ultrametric.  By Theorem \ref{thm:charact}, $\rev(g)$ is \unp. 
  Hence, Item (4.) and (5.) are equivalent.  
  \qed
\end{proof}

\section{Algorithms and Complexity Results}
\label{sec:alg}
\subsection{Recognition algorithm}

We first consider the problem of recognizing whether a 2-structure
$g=(V,\Upsilon,\vp)$ has a tree-representation $(T_g,t_g)$ and thus,
whether $g$ is \unp.  In what follows, the integer $n$ will always denote
$|V|$.  Based on Theorem \ref{thm:charactALL}(\ref{item:rev-g-unp}), we
will establish an $O(n^2)$ time algorithm to recognize \unp\ 2-structures.
We are aware of the existing $O(n^2)$ time algorithms described in
\cite{EHMS:94} and \cite{McConnell:95} to compute the modular decomposition
of 2-structures.  However, on the one hand, our established results allow
to design a conceptual simple algorithm by means of Theorem
\ref{thm:charactALL} and, on the other hand, the developed method might
provide an interesting starting point for further heuristics for
corresponding NP-complete editing problems, as we shall discuss later.

At first we compute the reversible refinement $rev(g)$ of $g$.
Then, for $rev(g)$ the monochromatic subgraphs $G_i(rev(g))$ are computed
and it is checked whether all $G_i(rev(g))$ are di-cographs or not.
If they are di-cographs then the respective 1-clusters $\mc{C}_i^1$ 
are extracted and it is checked if $\mc{C}^1$ forms a hierarchy.
Finally, if $\mc{C}^1$ is a hierarchy, then a tree can be constructed with
the method described in \cite{McConnell:05}.
By the following Lemma, constructing a tree from a hierarchy can be done in
linear time with respect to the number of elements in the hierarchy,
which, by Theorem \ref{A:thm:hierarchy}, is bounded by $2|V|-1$.
\begin{lemma}[\cite{McConnell:05}]
  \label{lem:inclusion-tree}
  Given a hierarchy $\mc{C}$, it takes $O(|\mc{C}|)$ time to construct its
  inclusion tree.
\end{lemma}
Pseudocode for the recognition procedure is given in Algorithm
\ref{alg:recog}.  Furthermore, we give pseudocode for all necessary
subroutines (Algorithms \ref{alg:rev-g} to \ref{alg:hierarchy}).  We omit
the procedure for computing the modular decomposition $\Ms(G)$ of a digraph
$G=(V,E)$ because McConnell and de Montgolfier \cite{McConnell:05} already
presented an $O(|V|+|E|)$ time algorithm for this problem.

We first prove the correctness of Algorithms \ref{alg:cograph},
\ref{alg:cluster}, and \ref{alg:hierarchy}.

\begin{lemma}
\label{lem:correct-cograph}
  Given a digraph $G$ and its modular decomposition $\Ms(G)$,
  Algorithm \ref{alg:cograph} recognizes whether
  $G$ is a di-cograph or not.
\end{lemma}
\begin{proof}
  At first, Algorithm \ref{alg:cograph} computes the inclusion tree 
  $T$ of $\Ms(G)$
  and then iterates over all strong modules $M \in \Ms(G)$.
  For each strong module $M$ two arbitrary but distinct children 
  $M',M'' \in \Ms(G)$ of $M$ in $T$
  are selected and it is checked if there is an arc between two 
  vertices $x \in M'$ and $y \in M''$.
  If $G$ is a di-cograph and there is an arc $(x,y) \in E$ or $(y,x) \in E$,
  then by Remark \ref{rem:cluster-ser-ord}, $M$ must be either series or order.
  In other words, if we have found an $(x,y) \in E$ or $(y,x) \in E$,
  but $M$ is neither series nor order, it must be prime which
  implies that $G$ was not a di-cograph. 
  However, it might be possible, that the chosen elements $x$ and $y$ 
  do not form an arc $(x,y) \in E$ or $(y,x) \in E$, but then $M$ is 
  either prime or parallel.
  If $M$ is prime there must be arcs $(x',y')$ or $(y',x')$,
  that we might have not observed in the preceding step, where $x'\in M'$, 
  $y'\in M''$ for some children $M',M''$ of $M$, 
  otherwise $M$ would be parallel.
  However, this case is covered by counting the numbers of all arcs between 
  the vertices of maximal strong submodules contained in series or order 
  modules $M$.
  If the accumulated number $e$ of all counted arcs is equal to the number
  of arcs $|E|$ in $G$, then all modules $M' \in \Ms(G)$ which are neither
  series nor order must be parallel. Hence, no prime modules exists and
  therefore $G$ is a di-cograph.  
  \qed
\end{proof}

\begin{lemma}
  \label{lem:correct-cluster}
  Given a di-cograph $G_i$ and its modular decomposition $\Ms(G_i)$,
  Algorithm \ref{alg:cluster} computes the 1-clusters $\mc{C}_i^1$ of $G_i$.
\end{lemma}
\begin{proof}
  At first, Algorithm \ref{alg:cluster} computes the inclusion tree $T$ of
  $\Ms(G_i)$.  Then, for each strong module $M$ two arbitrary vertices from
  distinct children $M',M'' \in \Ms(G)$ of $M$ in $T$ are selected.  If
  there is an arc $(x,y) \in E$ or $(y,x) \in E$, then by Remark
  \ref{rem:cluster-ser-ord}, $M$ cannot be parallel and hence, $M$ is a
  1-cluster and therefore, has to be added to the set of 1-clusters
  $\mc{C}_i^1$.  
  \qed
\end{proof}

\begin{algorithm}[tbp]
\caption{Recognition of \unp\ 2-Structures} 
\label{alg:recog}
\begin{algorithmic}[1]
\small
\STATE \textbf{INPUT:} 2-structure $g = (V,\Upsilon,\vp)$ with 
        $n=|V|$ vertices and $k=|\Upsilon|$ labels;
\STATE $g'=(V,\Upsilon',\vp') \gets$ \texttt{Compute $\rev(g)$}
\IF{$|\Upsilon'|>2(n-1)$}
	\RETURN \emph{FALSE}
\ENDIF

\STATE set of digraphs $\mathcal G \gets$ \texttt{Compute monochromatic subgraphs $G_i(g')$}
\STATE multiset of clusters $\mc{C} \gets \emptyset$
\FOR{$G_i$ in $\mathcal G$}
	\STATE set of strong modules $\Ms \gets$ \texttt{Compute the modular decomposition of $G_i$} (cf.\ \cite{McConnell:05})
	\IF{\texttt{Check di-cograph property for $G_i$ with modular decomposition $\Ms(G_i)$}}
		\STATE $\mc{C}_i^1 \gets$ \texttt{Get 1-clusters from di-cograph $G_i$ with modular decomposition $\Ms(G_i)$}
	\ELSE
		\RETURN \emph{FALSE}
	\ENDIF
	\STATE $\mc{C} \gets \biguplus_{i \in \Upsilon'} \mc{C}_i^1$
\ENDFOR

\IF{$|\mc{C}|>2(n-1)$}
	\RETURN \emph{FALSE}
\ENDIF
\IF {\texttt{Check hierarchy property for $\mc{C}$}}
	\RETURN \emph{TRUE}
\ELSE
	\RETURN \emph{FALSE}
\ENDIF
\end{algorithmic}
\end{algorithm}

\begin{algorithm}[t]
\caption{Compute rev(g)} 
\label{alg:rev-g}
\begin{algorithmic}[1]
\small
\STATE \textbf{INPUT:} 2-structure $g = (V,\Upsilon,\vp)$ with $n=|V|$ vertices;
\STATE $\Upsilon' \gets \emptyset$
\FOR{$i=1,\dots, n$}
	\FOR{$j=1,\dots, n$}
		\STATE $\vp'(i,j) \gets (\vp(i,j),\vp(j,i))$
		\STATE $\Upsilon' \gets \Upsilon' \cup \{\vp'(i,j)\}$
   	\ENDFOR
\ENDFOR
\RETURN $g'=(V,\Upsilon',\vp')$
\end{algorithmic}
\end{algorithm}

\begin{algorithm}[t]
\caption{Compute monochromatic subgraphs $G_i(g')$} 
\label{alg:subgraphs}
\begin{algorithmic}[1]
\small
\STATE \textbf{INPUT:} 2-structure $g' = (V,\Upsilon',\vp')$ with $n=|V|$ vertices  and $k'=|\Upsilon'|$ labels;
\STATE define bijection $\mu : \Upsilon' \rightarrow {1 \ldots k'}$
\STATE $\mathcal G \gets \emptyset$
\FOR{$i=1,\dots, k'$}
	\STATE $E_i \gets \emptyset$
	\STATE $G_i = (V,E_i)$
\ENDFOR
\FOR{$i=1,\dots, n$}
	\FOR{$j=1,\dots, n$}
		\STATE $E_{\mu(\vp'(i,j))} \gets E_{\mu(\vp'(i,j))} \cup \{(i,j)\}$
   	\ENDFOR
\ENDFOR
\RETURN $\mathcal G = \bigcup_{i=1}^{k'} \{G_i\}$
\end{algorithmic}
\end{algorithm}

\begin{algorithm}[t]
\caption{Check di-cograph property for $G_i$ with modular decomposition $\Ms(G_i)$} 
\label{alg:cograph}
\begin{algorithmic}[1]
\small
\STATE \textbf{INPUT:} digraph $G_i$ and its modular decomposition $\Ms(G_i)$;
\STATE tree $T \gets$ inclusion tree of $\Ms$
\STATE arc counter $e \gets 0$
\FOR{$M \in V(T)$}
	\STATE $M', M'' \gets$ two arbitrary but distinct child vertices of $M$ in $T$
	\STATE $x, y \gets$ two arbitrary elements $x \in M'$ and $y \in M''$
	\IF{$(x,y) \in E(G_i)$ or $(y,x) \in E(G_i)$}
		\IF{$M$ is not series or order}
			\RETURN \emph{FALSE}
		\ELSE
			\STATE increase $e$ by the number of arcs between all elements from distinct children of $M$ in $T$
		\ENDIF
	\ENDIF 
\ENDFOR
\IF{$e \neq  |E(G_i)|$}
	\RETURN \emph{FALSE}
\ENDIF
\RETURN \emph{TRUE}
\end{algorithmic}
\end{algorithm}

\begin{algorithm}[t]
\caption{Get 1-clusters from di-cograph $G_i$ with modular decomposition $\Ms(G_i)$} 
\label{alg:cluster}
\begin{algorithmic}[1]
\small
\STATE \textbf{INPUT:} di-cograph $G_i$ and its modular decomposition $\Ms(G_i)$;
\STATE $\mc{C}_i^1 \gets \emptyset$
\STATE tree $T \gets$ inclusion tree of $\Ms(G_i)$
\FOR{$M \in V(T)$}
	\STATE $M', M'' \gets$ two arbitrary but distinct child vertices of $M$ in $T$
	\STATE $x, y \gets$ two arbitrary elements $x \in M'$ and $y \in M''$
	\IF{$(x,y) \in E(G_i)$ or $(y,x) \in E(G_i)$}
		\STATE $\mc{C}_i^1 \gets \mc{C}_i^1 \cup M$
	\ENDIF
\ENDFOR
\RETURN $\mc{C}_i^1$
\end{algorithmic}
\end{algorithm}

\begin{algorithm}[htbp]
\caption{Check hierarchy property for $\mc{C}$} 
\label{alg:hierarchy}
\begin{algorithmic}[1]
\small
\STATE \textbf{INPUT:} multiset of clusters $\mc{C}$, on the ground set $\{1, \ldots, n\}$;
\STATE for each element in $\mc{C}$ compute a unique identifier $id : \mc{C} \rightarrow \{1, \ldots, |\mc{C}|\}$
\STATE for each element in $\mc{C}$ compute its bit string representation $bsr : \mc{C} \rightarrow \{0,1\}^n$ with $bsr(C_j)[i]=1$ iff $i \in C_j$, $C_j \in \mc{C}$
\STATE sorted list $\mc{C}_\leq \gets$ sort $\mc{C}$ ascending by cardinality of its elements
\FOR{i=1, \ldots, n}
	\IF{$bsr(\mc{C}_\leq(|\mc{C}_\leq|))[i] \neq 1$}
		\RETURN \emph{FALSE}
	\ENDIF
	\STATE $\mathcal L_i \gets \mc{C}_\leq$
	\FOR{$C_j \in \mathcal L_i$}
		\IF{$bsr(C_j)[i]=0$}
			\STATE remove $C_j$ from $\mathcal L_i$
		\ENDIF
	\ENDFOR
\ENDFOR
\WHILE{$\mathcal L_1 \neq \emptyset$}
	\STATE $s \gets$ the smallest $i$ such that $|\mathcal L_i(1)| \leq |\mathcal L_j(1)|$ for all $i \neq j$
	\STATE $L \gets \mathcal L_s(1)$
	\FOR{$t \in L$ with $t \neq s, \mathcal L_t \neq \emptyset$}
		\FOR{$r=1, \ldots, |\mathcal L_s|$}
			\IF{$id(\mathcal L_s(r)) \neq id(\mathcal L_t(r))$}
				\RETURN \emph{FALSE}
			\ENDIF
			\STATE $\mathcal L_t \gets \emptyset$
		\ENDFOR
	\ENDFOR
	\STATE remove $L$ from $ \mathcal L_s$
\ENDWHILE
\RETURN \emph{TRUE}
\end{algorithmic}
\end{algorithm}

The next lemma shows that Algorithm \ref{alg:hierarchy} correctly
recognizes, whether $\mc{C}^1(\rev(g))$ is a hierarchy or not.  However,
due to efficiency and also simplicity of the algorithm, we deal here with
multisets, $\mc{C} =\biguplus_{i \in \Upsilon_{rev(g)}} \mc{C}_i^1$.  The
symbol ``$\biguplus$'' denotes the multiset-union of sets where the
multiplicity of an element $M$ in $\mc{C}$ is given by the number of sets
that contain $M$.

\begin{lemma}
  \label{lem:correct-hirarchy}
  Given a multiset $\mc{C} =\biguplus_{i \in \Upsilon_{rev(g)}} \mc{C}_i^1$
  of the 1-clusters of a set of di-cographs $G_i=(V,E_i)$, Algorithm
  \ref{alg:hierarchy} recognizes whether $\mc{C}^1=\bigcup_{i \in
    \Upsilon_{rev(g)}} \mc{C}_i \cup \{v | v \in V\}$ is a hierarchy or
  not.
\end{lemma}
\begin{proof}
  Note that the multiset $\mc{C}$ may contain a cluster $C$ more than once,
  as $C$ can be part of different 1-clusters $\mc{C}_i^1$. Furthermore,
  $\mc{C}$ does not contain the singletons.  However, it is easy to see
  that $\mc{C}^1$ is a hierarchy if and only if the singletons are
  contained in $\mc{C}^1$ (which is satisfied by construction), there is a
  1-cluster equal to $V$ and for all $C', C'' \in \mc{C}$ it holds that $C'
  \cap C'' \in \{C',C'',\emptyset\}$.  The latter is equivalent to the
  following statement.  For all $C', C'' \in \mc{C}$, $|C'| \leq |C''|$ it
  holds that either $C' \cap C'' = \emptyset$ or $C' \subseteq C''$.
  
  In Line 4, a list $\mc{C}_\leq$ is created with all
  $C \in \mc{C}$ being sorted ascending by cardinality.
  Hence, $\mc{C}_\leq(|\mc{C}_\leq|)$ is one of the largest clusters.
  In Line 6, it is checked if this largest cluster contains all elements
  from the ground set $V=\{1, \ldots, n\}$. If not then $V \notin \mc{C}$
  and therefore $\mc{C}^1$ is not a hierarchy.
  In Lines 9 to 14, 	lists $\mathcal L_i$ are created,
  containing all clusters $C \in \mc{C}$ with
  $i \in C$.
  The relative order of clusters in $\mathcal L_i$
  is identical to the relative order of clusters in $\mc{C}_\leq$.
  In each iteration of Lines 16 to 28 the smallest cluster $L$ is selected  
  among all remaining clusters $\bigcup_{i=1}^n \mathcal L_i$.
  For each $i \in L$ obviously $L \in \mathcal L_i$.
  If $s,t \in L$ then it is checked if $\mathcal L_s = \mathcal L_t$.
  This can be done, as $\mathcal L_s$ and $\mathcal L_t$ have the same
  relative order of clusters.
  If $s,t \in L$ and $\mathcal L_s = \mathcal L_t$ then it follows that
  $s,t \in L'$ for all $L' \in \mathcal L_s \cup \mathcal L_t$.
  As this holds for all pairwise distinct $s,t \in L$
  and $|L| \leq |L'|$ for all $L' \in \bigcup_{i=1}^n \mathcal L_i$
  it follows that $L \subseteq L'$ for all 
  $L' \in \bigcup_{i=1}^n \mathcal L_i$ with $L \cap L' \neq \emptyset$.
  As $\mathcal L_s = \mathcal L_t$ it is sufficient to keep only one
  of the lists, e.g., $\mathcal L_s$ (Line 24).
  Finally, $L$ is removed from $\mathcal L_s$ (Line 27) and
  the while-loop is repeated with the next smallest cluster.
\qed\end{proof}

We now show the correctness of Algorithm \ref{alg:recog}.
\begin{lemma}
  \label{lem:correct-recog}
  Given a 2-structure $g = (V,\Upsilon,\vp)$, Algorithm \ref{alg:recog}
  recognizes whether $g$ is \unp\ or not.
\end{lemma}
\begin{proof}
  In fact, Algorithm \ref{alg:recog} recognizes, for the  reversible
  refinement $rev(g)$, whether all monochromatic subgraphs 
  $G_i(rev(g))$ are di-cographs and whether in addition the 1-clusters 
  in $\mc{C}^1(rev(g))$ form a hierarchy. By Theorem \ref{thm:charactALL}, 
  this suffices to decide whether $g$ is \unp\ or not.

  It is easy to see that Algorithm \ref{alg:rev-g} computes the reversible
  refinement of $g$ by means of Definition \ref{def:rev-g} with
  $h(\vp_g(e),\vp_g(e^{-1})) = (\vp_g(e),\vp_g(e^{-1}))$.  Hence, in Line 2
  the reversible refinement $g'=\rev(g)$ of $g$ is computed.
	
  If $\rev(g)$ is \unp,
  then there exists a tree representation $(T_{\rev(g)},t_{\rev(g)})$.
  As $T_{\rev(g)}$ has at most $n-1$ inner vertices
  there can be at most $n-1$ different labels $t_{\rev(g)}(lca(x,y))=(i,j)$,
  each composed of at most two distinct labels $i,j \in \Upsilon_{rev(g)}$.
  Assuming that all labels are pairwise distinct leads to 
  $2(n-1) \leq |\Upsilon_{rev(g)}|$ distinct labels in total.
  Hence, if $|\Upsilon_{rev(g)}| > 2(n-1)$ then $rev(g)$ is not \unp.
  It is easy to see that, given the 2-structure $\rev(g)$, 
  Algorithm \ref{alg:subgraphs} (which is called in Line 6) 
  computes the respective
  monochromatic subgraphs $G_i(\rev(g))$.
  By Lemma \ref{lem:correct-cograph}, for each $G_i$ Algorithm \ref{alg:cograph}
  (which is called in Line 10) checks whether $G_i$ is a di-cograph or not, 
  and by Lemma \ref{lem:correct-cluster} in Line 11
  the corresponding 1-clusters $C_i^1$ are returned.
  In Line 15, the 1-clusters $\mc{C}_i^1$ of all di-cographs $G_i$
  are collectively stored in the multiset $\mc{C}$, without removing
  duplicated entries.
 
  Since $T_{\rev(g)}$ has at most most $n-1$ inner vertices and since
  each 1-cluster appears in at most 2 distinct cotrees whenever $\rev(g)$
  is \unp\ (cf.\ Lemma \ref{lem:unp-1clust_new}), we can conclude that $\mc{C}$
  can contain at most $2(n-1)$ elements. Hence, if 
  $|\mc{C}| > 2(n-1)$ then $\rev(g)$ is not \unp, and therefore, 
  $g$ is not \unp\ (Line 17).

  Finally, by Lemma \ref{lem:correct-hirarchy} it is checked in Line 20, if
  the set of 1-clusters $\mc{C}^1$ is a hierarchy.  Hence, \emph{TRUE} is
  returned if $g$ is \unp\ and \emph{FALSE} else.  
  \qed
\end{proof}

Before we show the time complexity of Algorithm \ref{alg:recog} we first
show the time complexity of the two subroutines Algorithm \ref{alg:cograph}
and Algorithm \ref{alg:hierarchy}.

\begin{lemma}
  \label{lem:complex-cograph}
  For a given digraph $G=(V,E)$ and its modular decomposition $\Ms$,
  Algorithm \ref{alg:cograph} runs in time $O(n+m)$ with $n=|V|$ and
  $m=|E|$.
\end{lemma}
\begin{proof}
  By Lemma \ref{lem:inclusion-tree} computing
  the inclusion tree $T$ of $\Ms(G)$ in Line 2 takes time $O(n)$
  as there are at most $O(n)$ strong modules.
  In the for-loop from Line 4 to 14 for each strong module $M$ it is checked,
  whether or not there is an arc between two arbitrary vertices
  from two distinct children of $M$ in $T$.
  This has to be done for all $O(n)$ strong modules $M \in \Ms(G)$.
  Only if there is an arc it is further checked whether $M$ is
  series or order.
  This can be done by checking all the arcs between vertices $x$ and $y$
  from distinct children of $M$ in $T$. 
  In both cases ($M$ being series and order)
  there is at least one arc $(x,y) \in E$ or $(y,x) \in E$,
  between any pair of vertices $x$ and $y$.
  Furthermore, as only vertices from distinct children of $M$ in $T$
  are considered, every pair $(x,y)$ is checked at most once once.
  Hence, the number of all pairwise checks is bounded by $O(m)$.
  For the same reason, counting the arcs (Line 11)
  can also be done in $O(m)$ time.
  This accounts to a running time of $O(n+m)$ in total. 
  \qed
\end{proof}

\begin{lemma}
  \label{lem:complex-hirarchy}
  For a given multiset of clusters $\mc{C}$ of size $N$
  on the ground set $\{1, \ldots, n\}$, 
  Algorithm \ref{alg:hierarchy} runs in time $O(n N)$.
\end{lemma}
\begin{proof}
  Computing the identifier $id$ for each cluster in $\mc{C}$ (Line 2) takes
  time $O(N)$, computing the bit string representation for each cluster in
  $\mc{C}$ (Line 3) takes time $O(n N)$, and sorting the clusters of
  $\mc{C}$ (Line 4) using bucket sort with $n$ buckets takes time $O(N+n)$.
  The for-loop from Line 5 to Line 15 runs in time $O(n N)$, as there are
  $O(N)$ clusters in $\mc{C}$ which possibly have to be removed in Line 12
  from the respective lists $\mathcal L_i$.  The while-loop (Lines 16 to
  28) is executed at most $O(N)$ times, as in each iteration one of the $N$
  clusters is removed from all the lists $\mathcal L_i$ that contain it
  (Line 24 and 27).  The for-loop from Line 19 to Line 26 is executed for
  all of the $O(n)$ many elements $t \in L$.  However, as in each execution
  of the inner loop (Lines 20 to 25) one of the $n$ lists $\mathcal L_i$
  gets empty, Lines 20 to 25 are executed $n$ times in total and each
  execution takes $O(N)$ time.  Hence, the time that Algorithm
  \ref{alg:hierarchy} spends on computing Lines 20 to 25 is bounded by
  $O(n N)$.  This sums up to a total running time of $O(n N)$ for Algorithm
  \ref{alg:hierarchy}.  \qed\end{proof}

Finally, we show the time complexity of $O(n^2)$ for Algorithm \ref{alg:recog}.

\begin{lemma}
\label{lem:complex-recog}
  For a given 2-structure $g=(V,\Upsilon,\vp)$ with $n=|V|$, Algorithm \ref{alg:recog} runs in time $O(n^2)$.
\end{lemma}
\begin{proof}
  Computing the reversible refinement of $g$ in Line 2
  takes $O(n^2)$ time using Algorithm \ref{alg:rev-g}.
  In Line 3 it is assured that there are at most $2(n-1)$ labels
  and hence $N = |\Upsilon_{rev(g)}| < 2(n-1)$ monochromatic subgraphs 
  $G_i(rev(g))$.
  Computing those $O(n)$ subgraphs at once
  using Algorithm \ref{alg:subgraphs}
  in Line 6 takes $O(n^2)$ time.
  The for-loop from Line 8 to Line 16 runs for each of the $O(n)$
  many digraphs $G_i(rev(g))$.
  As already stated, there is an $O(n+m)$ time complexity algorithm for
  computing the modular decomposition of a digraph (Line 9) given in
  \cite{McConnell:05}. 	
  By Lemma \ref{lem:complex-cograph} Algorithm \ref{alg:cograph} (Line 10)
  has also a time complexity of $O(n+m)$.
  Algorithm \ref{alg:cluster} (Line 11) has a time complexity of $O(n)$,
  as by Lemma \ref{lem:inclusion-tree}
  constructing the inclusion tree within Line 3 of Algorithm \ref{alg:cluster}
  takes time $O(n)$ as there are at most $O(n)$ strong modules within $G_i$.
  Hence, all procedures within the for-loop (Lines 8 to 16) have a time 
  complexity of $O(n+m)$. Precisely, the time complexity is 
  $O(n+m_i)$ with $m_i=|E(G_i(rev(g)))|$ the number of arcs of $G_i(rev(g))$.
  The total running time of the for-loop therefore is
  $O(n+m_1) + O(n+m_2) + \ldots O(n+m_N) = O(n^2 + \sum_{i=1}^N m_i)$.
  As each arc $(x,y)$ occurs	in exactly one of the
  digraphs $G_i(rev(g))$ it follows that $\sum_{i=1}^N m_i = n(n-1)$,
  which leads to a running time of $O(n^2)$ for Line 8 to 16.
  Line 17 assures that the multiset $\mc{C}$ contains at most $2(n-1)$ clusters.
  Hence, $|\mc{C}| \in O(n)$.
  Therefore, and by Lemma \ref{lem:complex-hirarchy}
  Algorithm \ref{alg:hierarchy} runs in time $O(n^2)$.
  This leads to a time complexity of $O(n^2)$ for Algorithm \ref{alg:recog}.
  \qed
\end{proof}

We finish this section, by showing how the tree-representation of a \unp\
2-structure $g=(V,\Upsilon, \vp)$ can be computed.  By the preceding
results, we can compute the hierarchy $\mc{C}^1(\rev(g))$ in $O(|V|^2)$
time.  By Lemma \ref{lem:unp-1clust_new},
$\mc{C}^1(\rev(g))=\Ms(\rev(g))$. By Theorem
\ref{thm:points}(\ref{item:same-modules}), $\M(g)=\M(\rev(g))$ and hence,
$\Ms(\rev(g))=\Ms(g)$, which implies that $\mc{C}^1(\rev(g))=\mc{C}^1(g)$.
By Theorem \ref{A:thm:hierarchy}, the number of clusters contained in
$\mc{C}^1(g)$ is bounded by $2|V| - 1$.  Thus, we can compute the inclusion
tree $T_g$ of $\mc{C}^1(g)$ in $O(|V|)$ time by means of Lemma
\ref{lem:inclusion-tree}. In order to get the correct labeling $t_g$ we
proceed as follows.
We traverse $T_g$ via breadth-first search, starting with the root $v$
of $T_g$ that represents $V$. Take any two children $u_1, u_2$ of $v$
and any two vertices $x\in L(u_1)$, $y\in L(u_2)$ and check the labeling
of the arcs $(xy)$ and $(yx)$. Assume that $\vp(xy)=i$ and 
$\vp(yx)=j$, where $i\leq j$. Then set $t_g(v)=(i,j)$ and place
$u_1$ left of $u_2$ in $T_g$. An ordering of the children of $v$
is not necessary if $i=j$. This step has to be repeated for all
pairs of children of $v$ and thus has
total time complexity of $O(\deg(v)^2)$, where 
$\deg(v)$ denotes the number of children of $v$. 
After this we proceed with a child of $v$, playing now the role of $v$. 
Hence, the ordering of the tree $T_g$ and the labeling $t_g$ can be computed in 
$\sum_{v \in V(T_g)} \deg(v)^2 \leq \sum_{v \in V(T_g)} |V|\deg(v) = |V|\cdot 2(|V|-1)$ and 
thus, in $O(|V|^2)$ time. 

Taken the latter together with the preceding results, we obtain the
following main result of this section.
\begin{theorem}
  For a given 2-structure $g=(V,\Upsilon, \vp)$ it can be verified in 
  $O(|V|^2)$ time whether $g$ is \unp\ or not, and in the positive
  case, the tree representation $(T_g,t_g)$ can be computed in 
  $O(|V|^2)$ time.
  \label{thm:algALL}
\end{theorem}

\subsection{Tree-representable Sets of Relations, Complexity Results and ILP}

From the practical point of view, 2-Structures $g=(V,\Upsilon,\vp)$ with
$\Upsilon=\{0,1,\dots,k\}$ can be used to represent sets of disjoint
relations $R_1,\dots,R_k$, that is, $\vp$ is chosen s.t.\
\begin{align*}
  \vp(xy) = \begin{cases}
    i\neq 0,  & \text{ if } (x,y)\in R_i\\
    0 & \text{ else, i.e., } (x,y)\not\in R_i, 1\leq i\leq k
  \end{cases}
\end{align*}

Moreover, 2-structures can even be used to represent \emph{non-disjoint}
relations as follows.  Assume that we have arbitrary binary relations
$R_1,\dots,R_k$ over some set $V$. For two vertices $x,y\in V$ let
$I_{xy}\subseteq \{1,\dots,k\}$ be the inclusion-maximal subset s.t.\
$(x,y)\in R_i$ for all $i\in I_{xy}$. Furthermore, let $b_{xy}$ be the
(unique) integer encoded by the bit vector $b_1b_2\dots b_k$ with $b_i=1$
iff $i\in I_{xy}$.  Clearly, $g=(V,\mathbb{N},\vp)$ with $\vp(xy)=b_{xy}$
is a 2-structure, and the disjoint relations can be represented in a tree
if and only if $g$ is \unp.

Those relations might represent the evolutionary relationships between
genes, i.e., genes $x,y$ are in relation $R_i$ if the lowest common
ancestor $\lca(x,y)$ in the corresponding gene tree was labeled with a
particular event $i$, as e.g.  speciation, duplication, horizontal-gene
transfer, retro-transposition, and others.  By way of example, methods as
\texttt{ProteinOrtho} \cite{Lechner:11a,Lechner:14} allow to estimate pairs
of orthologs without inferring a gene or species tree.  Hence, in practice
such relations represent often only estimates of the true relationship
between genes.  Thus, in general the 2-structure of such estimates
$R_1\dots,R_k$ will not be \unp\ and hence, there is no tree-representation
of such estimates. One possibility to attack this problem is to optimally
edit the estimate $g = \{V,\Upsilon, \vp\}$ to a \unp\ 2-structure $g^* =
\{V,\Upsilon, \vp^*\}$ by changing the minimum number of colors assigned by
$\vp$.

We first consider the problem to rearrange a symmetric map $d$ to obtain a
symmetric symbolic ultrametric $\delta$.

\begin{problem}\textsc{Symmetric Symbolic Ultrametric  
    Editing/Deletion/Completion [SymSU-E/D/C]}
\vspace*{0.3em}
\begin{tabular}{lll}
  \emph{Input:}
  &\multicolumn{2}{l}{Given a symmetric map $d : \Virr \to \Upsilon$, 
    a fixed symbol $\ast\in \Upsilon$ and }\\
  & an integer $k$. & \\ 
  \emph{Question} & \multicolumn{2}{l}{Is there a symmetric symbolic 
    ultrametric $\delta : \Virr \to \Upsilon$, s.t.\ } \\[0.1cm]
  & $\bullet$ $|D| \leq k$ & \emph{(Editing)} \\ 
  & $\bullet$ if $d(x,y)\neq \ast$, 
    then $\delta(x,y)=d(x,y)$; and  $|D| \leq k$ &\emph{(Completion)}\\
  & $\bullet$ $\delta(x,y)=d(x,y)$ or $\delta(x,y) = \ast$; and $|D|\leq k$ 
  &\emph{(Deletion)}\\[0.1cm]
  &\multicolumn{2}{l}{where $D = \{(x,y)\in X\times X \mid d(x,y) 
                      \neq \delta(x,y)\}$. } 				
\end{tabular}
\end{problem}

The editing problem is clear. The completion problem is motivated by
assuming that we might have an reliable assignment $d$ on a subset $W$ of
$\Virr$, however, the assignment for the pairs $(x,y)\in \overline{W} =
\Virr \setminus W$ is unreliable or even unknown. For those pairs we use an
extra symbol $\ast$ and set $d(x, y)=\ast$ for all $(x,y)\in \overline{W}$.
Since we trust in the assignment $d$ for all elements in $W$ we aim at
changing the least number of non-reliable estimates only, that is, only
pairs $(x,y)$ with $d(x,y)=\ast$ are allowed to be changed so that the
resulting map becomes a symmetric symbolic ultrametric.  Conversely, the
deletion problem asks to change a minimum number of assignments $d(x, y)
\neq \ast$ to $\delta(x, y) = \ast$.

The following result was given in \cite{HW:16}. 
\begin{theorem}
\textsc{SymSU-E/D/C} is NP-complete.
\end{theorem}

We can use this result in order to show that the analogous problems for
2-structures are NP-complete, as well.

\begin{problem}\textsc{\unp\ 2-Structure  Editing/Deletion/Completion [u2s-E/D/C]  \\}
\vspace{0.3em}
\begin{tabular}{lll}
	\emph{Input:}&  \multicolumn{2}{l}{Given a 2-structure $g= (V, \Upsilon, \vp_g)$, a fixed symbol $\ast\in \Upsilon$ and }\\
				& an integer $k$. & \\ 
	\emph{Question} & \multicolumn{2}{l}{Is there a \unp\ 2-structure $h= (V, \Upsilon, \vp_h)$, s.t.\ } \\[0.1cm]
					 & $\bullet$  $|D| \leq k$ & \emph{(Editing)} \\ 
					 & $\bullet$ if $\vp_g(x,y)\neq \ast$, then $\vp_h(x,y)=\vp_g(x,y)$; and  $|D| \leq k$ &\emph{(Completion)}\\
						& $\bullet$ $\vp_h(x,y)=\vp_g(x,y)$ or $\vp_h(x,y) = \ast$; and  $|D| \leq k$ &\emph{(Deletion)}\\[0.1cm] 
	&\multicolumn{2}{l}{where $D = \{(x,y)\in X\times X \mid \vp_g(x,y) \neq \vp_h(x,y)\}$.} 
				
\end{tabular}
\end{problem}

\begin{theorem}
\textsc{ u2s-E/D/C } is NP-complete.
\label{thm:npc}
\end{theorem}
\begin{proof}
  Since we can test whether a 2-structure is \unp\ with Algorithm 
  \ref{alg:recog} in polynomial time, \textsc{ u2s-E/D/C } $\in NP$.
  
  To show NP-hardness, we simply reduce the instance $d:\Virr\to\Upsilon$
  of \textsc{SymSU} to the instance $g= (V, \Upsilon, d)$ of
  \textsc{u2s}. By Theorem \ref{thm:charact}, $\delta$ is a symbolic
  ultrametric if and only if $h= (V, \Upsilon, \delta)$ is \unp\ and thus,
  we obtain the NP-hardness of \textsc{u2s-E/D/C}.  
  \qed
\end{proof}

The latter proof in particular implies that \textsc{u2s-E/D/C} is even
NP-complete in the case that $\vp_g(xy)=\vp_g(yx)$ for all distinct $x,y\in
V$.

We showed in \cite{Hellmuth:15a} that the cograph editing problem and in
\cite{HW:16} that the symmetric symbolic ultrametric
editing/completion/deletion is amenable to formulations as Integer Linear
Program (ILP). We will extend these results here to solve the symbolic
ultrametric editing/completion/deletion problem.

Let $d:\Virr\to \Upsilon$ be an arbitrary map with
$\Upsilon=\{\ast,1,\ldots,n\}$ and $K_{|V|}=(V,E=\Virr)$ be the
corresponding complete di-graph with arc-coloring s.t.\ each arc $(x,y)\in
E$ obtains color $d(x,y)$.

For each of the three problems and hence, a given symmetric map $d$ we
define for each distinct $x,y\in V$ and $i\in \Upsilon$ the binary
constants $\mathfrak{d}^i_{x,y}$ with $\mathfrak{d}^i_{x,y}=1$ if and only
if $d(x,y)=i$.  Moreover, we define the binary variables $E^{ij}_{xy}$ for
all $i,j \in \Upsilon$ and $x,y \in V$ that reflect the coloring of the
arcs in $K_{|V|}$ of the final symbolic ultrametric $\delta$, i.e.,
$E^{ij}_{xy}$ is set to $1$ if and only if $\delta(x,y)=i$ and
$\delta(y,x)=j$.  In the following we will write $E^i_{xy}$ as a shortcut
for $\sum_{j \in \Upsilon}E^{ij}_{xy}$.  Note that $E^i_{xy} \in \{0,1\}$
and $E^i_{xy}=1$ if and only if $\delta(x,y)=i$.

In order to find the closest symbolic ultrametric $\delta$, 
the objective function is to minimize the symmetric difference of
$d$ and $\delta$ among all different symbols $i\in \Upsilon$:
\begin{equation}
\min \sum_{i \in \Upsilon} \Bigg( \sum_{(x,y) \in \Virr}
(1-\mathfrak{d}^i_{xy}) E^{i}_{xy} 
+ \sum_{(x,y) \in \Virr} \mathfrak{d}^i_{xy}(1-E^{i}_{xy}) \Bigg)
\end{equation}
The same objective function can be used for
the symbolic ultrametric completion and deletion problem. 

\smallskip
For the symbolic ultrametric completion
we must ensure that $\delta(x,y)=d(x,y)$ for all $d(x,y)\neq \ast$.
Hence we set for all $x,y$ with $d(x,y)=i\neq \ast$:
\begin{align}
E^{i}_{xy} =1. 
 \label{eee1}
\end{align}

For the symbolic ultrametric deletion we must ensure that
$\delta(x,y)=d(x,y)$ or $\delta(x,y)= \ast$.  In other words, for all
$d(x,y)=i\neq \ast$ it must hold that for some $j \in \Upsilon$ either
$E^{ij}_{xy}=1$ or $E^{\ast j}_{xy}=1$.  Hence, we set for all $(x,y)\in
\Virr$:
\begin{align}
  E^{\ast}_{xy}=1, \mbox{ if } d(x,y) = \ast \text{, and }  
  E^{i}_{xy}+E^{\ast}_{xy}=1,  \text{ else.} \tag{\ref{eee1}'}
  \label{eee2}
\end{align}

For the cograph editing problem we neither need Constraint \ref{eee1} nor
\ref{eee2}.  However, for all three problems we need the following.

\smallskip 
Each tuple $(x,y)$ with $x \neq y$ has exactly one pair of values $(i,j)
\in \Upsilon \times \Upsilon$ assigned to it, such that $E^{ij}_{xy} =
E^{ji}_{yx}$. Hence, we add the following constraints for all distinct
$(x,y) \in \Virr$ and $(i,j) \in \Upsilon \times \Upsilon$.
\begin{equation}
  \sum_{i,j \in \Upsilon} E^{ij}_{xy}=1  \text{ and } E^{ij}_{xy} = E^{ji}_{yx}.
\end{equation}

In order to satisfy Condition (U2) and thus, that all induced triangles
have at most two color-pairs we need to add the following constraints:
\begin{equation}
E^{ij}_{xy}+E^{kl}_{yz}+E^{rs}_{zx}\leq 2
\end{equation}
for all (not necessarily distinct) colors $i,j,k,l,r,s\in \Upsilon$ with
pairwise distinct $\{i,j\}$, $\{k,l\}$, and $\{r,s\}$ and for all distinct
$x,y,z\in V$.

Finally, in order to satisfy Condition (U1) and thus, that each
mono-chromatic subgraph comprising all arcs with fixed color $i$ is a
di-cograph, we need the a couple of constraints that encode forbidden
subgraphs. Since these conditions are straightforward to derive, we just
give two example constraints to avoid induced $P_4$'s and $\overline{N}$'s.
\begin{equation} 
E^{i}_{ab} + E^{i}_{ba} + E^{i}_{bc} + E^{i}_{cb} + 
E^{i}_{cd} + E^{i}_{dc} - E^{i}_{ac} - E^{i}_{ca} - 
E^{i}_{ad} - E^{i}_{da} - E^{i}_{bd} - E^{i}_{db} \leq 5
\label{fsg1}
\end{equation}
\begin{equation}
E^{i}_{ac} + E^{i}_{ca} + E^{i}_{ad} + E^{i}_{da} + 
E^{i}_{bd} + E^{i}_{db} + E^{i}_{ba} + E^{i}_{bc} + 
E^{i}_{dc} - E^{i}_{ab} - E^{i}_{cb} - E^{i}_{bd} \leq 8
\tag{\ref{fsg1}'}
\label{fsg2}
\end{equation}
for all $i \in \Upsilon$ and all ordered tuple $(a,b,c,d)$ of distinct
$a,b,c,d \in V$.

It is easy to verify that the latter ILP formulation needs
$O(|\Upsilon|^2|V|^2)$ variables and $O(|\Upsilon|^6|V|^3+|\Upsilon||V|^4)$
constraints.

\section{Concluding Remarks} 

From an applications point of view, the main result of this contribution is
a classification of those relationships between genes that can be derived
from a gene phylogeny and the knowledge of event types assigned to interior
\emph{nodes} of phylogenetic tree.  All such relations necessarily have
co-graph structure. Fitch's version of the orthology and paralogy relations
are the most important special cases. However, this class of relations is
substantially more general and also include non-symmetric relations. These
can account in particular for pairs of genes that are related by an
ancestral horizontal transfer event and keep track of the directionality of
the transfer. It remains an open question to what extent this
``lca-xenology'' relation can be inferred directly from sequence similarity
data similar to the orthology and paralogy relations. 

The most commonly used definition of the xenology relation, however, is
based on the presence of one or more horizontal transfer events along the
unique path in the gene tree that connects two genes. It cannot be
expressed in terms labels at the lowest common ancestor only. This raises the
question whether edge labeled phylogenetic trees given rise to similar
systems of relations on the leaf set. 

\section*{Acknowledgements}
  We thanks Maribel Hern{\'a}ndez-Rosales for discussions.  This work was
  funded by the German Research Foundation (DFG) (Proj.\ No.\ MI439/14-1 to
  P.F.S. and N.W.).

\bibliographystyle{plain}
\bibliography{biblio}

\begin{thebibliography}{10}

\bibitem{BD98}
Sebastian B\"{o}cker and Andreas W.~M. Dress.
\newblock Recovering symbolically dated, rooted trees from symbolic
  ultrametrics.
\newblock {\em Advances in Mathematics}, 138:105--125, 1998.

\bibitem{BLS:99}
Andreas Brandst{\"a}dt, Van~Bang Le, and Jeremy~P Spinrad.
\newblock {\em Graph Classes: A Survey}.
\newblock Society for Industrial and Applied Mathematics, Philadelphia, PA,
  USA, 1999.

\bibitem{Corneil:81}
D.~G. Corneil, H.~Lerchs, and L.~Steward~Burlingham.
\newblock Complement reducible graphs.
\newblock {\em Discr. Appl. Math.}, 3:163--174, 1981.

\bibitem{CP-06}
C.~Crespelle and C.~Paul.
\newblock Fully dynamic recognition algorithm and certificate for directed
  cographs.
\newblock {\em Discr. Appl. Math.}, 154:1722--1741, 2006.

\bibitem{Ehrenfeucht1995}
A.~Ehrenfeucht, T.~Harju, and G.~Rozenberg.
\newblock Theory of 2-structures.
\newblock In Zolt{\'a}n F{\"u}l{\"o}p and Ferenc G{\'e}cseg, editors, {\em
  Automata, Languages and Programming: Proceedings of the 22nd International
  Colloquium, {ICALP} 95 Szeged, Hungary, July 10--14, 1995}, pages 1--14.
  Springer, Berlin, Heidelberg, 1995.

\bibitem{ER1:90}
A~Ehrenfeucht and G~Rozenberg.
\newblock Theory of 2-structures, part {I}: Clans, basic subclasses, and
  morphisms.
\newblock {\em Theor. Comp. Sci.}, 70:277--303, 1990.

\bibitem{ER2:90}
A~Ehrenfeucht and G~Rozenberg.
\newblock Theory of 2-structures, part {II}: Representation through labeled
  tree families.
\newblock {\em Theor. Comp. Sci.}, 70:305--342, 1990.

\bibitem{EHMS:94}
Andrzej Ehrenfeucht, Harold~N. Gabow, Ross~M. Mcconnell, and Stephen~J.
  Sullivan.
\newblock {An O($n^2$) Divide-and-Conquer Algorithm for the Prime Tree
  Decomposition of Two-Structures and Modular Decomposition of Graphs}.
\newblock {\em Journal of Algorithms}, 16(2):283--294, 1994.

\bibitem{ehrenfeucht1999theory}
Andrzej Ehrenfeucht, Tero Harju, and Grzegorz Rozenberg.
\newblock {\em The theory of 2-structures: A framework for decomposition and
  transformation of graphs}.
\newblock World Scientific, Singapore, 1999.

\bibitem{EHPR:96}
J.~Engelfriet, T.~Harju, A.~Proskurowski, and G~Rozenberg.
\newblock Characterization and complexity of uniformly nonprimitive labeled
  2-structures.
\newblock {\em Theor. Comp. Sci.}, 154:247--282, 1996.

\bibitem{Fitch:70}
W~M Fitch.
\newblock Distinguishing homologous from analogous proteins.
\newblock {\em Syst Zool}, 19:99--113, 1970.

\bibitem{Fitch:00}
W~M Fitch.
\newblock Homology a personal view on some of the problems.
\newblock {\em Trends Genet.}, 16:227--231, 2000.

\bibitem{Gray:83}
G~S Gray and W~M Fitch.
\newblock Evolution of antibiotic resistance genes: the {DNA} sequence of a
  kanamycin resistance gene from \emph{Staphylococcus aureus}.
\newblock {\em Mol Biol Evol}, 1:57--66, 1983.

\bibitem{hellmuth2015techniques}
Marc Hellmuth, Adrian Fritz, Nicolas Wieseke, and Peter~F Stadler.
\newblock Techniques for the cograph editing problem: Module merge is
  equivalent to editing p4s.
\newblock {\em arXiv preprint arXiv:1509.06983}, 2015.

\bibitem{Hellmuth:13a}
Marc Hellmuth, Maribel Hernandez-Rosales, Katharina~T. Huber, Vincent Moulton,
  Peter~F. Stadler, and Nicolas Wieseke.
\newblock Orthology relations, symbolic ultrametrics, and cographs.
\newblock {\em J. Math. Biol.}, 66:399--420, 2013.

\bibitem{HW:15}
Marc Hellmuth and Nicolas Wieseke.
\newblock On symbolic ultrametrics, cotree representations, and cograph edge
  decompositions and partitions.
\newblock In Dachuan Xu, Donglei Du, and Dingzhu Du, editors, {\em Computing
  and Combinatorics}, volume 9198 of {\em Lecture Notes in Computer Science},
  pages 609--623. Springer International Publishing, 2015.

\bibitem{HW:16b}
Marc Hellmuth and Nicolas Wieseke.
\newblock From sequence data incl. orthologs, paralogs, and xenologs to gene
  and species trees.
\newblock {\em arXiv preprint arXiv:1602.08268}, 2016.

\bibitem{HW:16}
Marc Hellmuth and Nicolas Wieseke.
\newblock On tree representations of relations and graphs: Symbolic
  ultrametrics and cograph edge decompositions.
\newblock {\em arXiv preprint arXiv:1509.05069}, 2016.

\bibitem{Hellmuth:15a}
Marc Hellmuth, Nicolas Wieseke, Marcus Lechner, Hans-Peter Lenhof, Martin
  Middendorf, and Peter~F. Stadler.
\newblock Phylogenomics with paralogs.
\newblock {\em Proc. Natl. Acad. Sci. USA}, 112:2058--2063, 2015.
\newblock doi: 10.1073/pnas.1412770112.

\bibitem{HernandezRosales:12a}
Maribel Hernandez-Rosales, Marc Hellmuth, Nick Wieseke, Katharina~T. Huber,
  Vincent Moulton, and Peter~F. Stadler.
\newblock From event-labeled gene trees to species trees.
\newblock {\em BMC Bioinformatics}, 13(Suppl.\ 19):S6, 2012.

\bibitem{Jensen:01}
Roy~A Jensen.
\newblock Orthologs and paralogs -- we need to get it right.
\newblock {\em Genome Biol}, 2:8, 2001.

\bibitem{Keeling:08}
Patrick~J. Keeling and Jeffrey~D. Palmer.
\newblock Horizontal gene transfer in eukaryotic evolution.
\newblock {\em Nature Rev Genetics}, 9:605--618, 2008.

\bibitem{Koonin:05}
Eugene Koonin.
\newblock Orthologs, paralogs, and evolutionary genomics.
\newblock {\em Ann. Rev. Genetics}, 39:309--338, 2005.

\bibitem{Koonin:01}
Eugene~V. Koonin, Kira~S. Makarova, and L.~Aravind.
\newblock Horizontal gene transfer in prokaryotes: Quantification and
  classification.
\newblock {\em Annu. Rev. Microbiol.}, 55:709--742, 2001.

\bibitem{Lechner:11a}
Marcus Lechner, Sven Findei{\ss}, Lydia Steiner, Manja Marz, Peter~F. Stadler,
  and Sonja~J. Prohaska.
\newblock \texttt{Proteinortho:} detection of (co-)orthologs in large-scale
  analysis.
\newblock {\em BMC Bioinformatics}, 12:124, 2011.

\bibitem{Lechner:14}
Marcus Lechner, Maribel Hernandez-Rosales, D.~Doerr, N.~Wiesecke, A.~Thevenin,
  J.~Stoye, Roland~K. Hartmann, Sonja~J. Prohaska, and Peter~F. Stadler.
\newblock Orthology detection combining clustering and synteny for very large
  datasets.
\newblock {\em PLoS ONE}, 9(8):e105015, 08 2014.

\bibitem{McConnell:95}
Ross~M. McConnell.
\newblock An $o(n^2)$ incremental algorithm for modular decomposition of graphs
  and 2-structures.
\newblock {\em Algorithmica}, 14(3):229--248, 1995.

\bibitem{McConnell:05}
Ross~M McConnell and Fabien de~Montgolfier.
\newblock Linear-time modular decomposition of directed graphs.
\newblock {\em Discrete Applied Mathematics}, 145(2):198--209, 2005.

\bibitem{Moehring:85}
R.~H. M{\"o}hring.
\newblock Algorithmic aspects of the substitution decomposition in optimization
  over relations, set systems and boolean functions.
\newblock {\em Annals of Operations Research}, 4(1):195--225, 1985.

\bibitem{Moehring:84}
R.~H. M{\"o}hring and F.~J. Radermacher.
\newblock Substitution decomposition for discrete structures and connections
  with combinatorial optimization.
\newblock {\em Ann. Discrete Math.}, 19:257--356, 1984.

\bibitem{sem-ste-03a}
Charles Semple and Mike Steel.
\newblock {\em Phylogenetics}, volume~24 of {\em Oxford Lecture Series in
  Mathematics and its Applications}.
\newblock Oxford University Press, Oxford, 2003.

\bibitem{Valdes:82}
J.~Valdes, R.~E. Tarjan, and E.~L. Lawler.
\newblock The recognition of series parallel digraphs.
\newblock {\em SIAM J. Computing}, 11:298--313, 1982.

\end{thebibliography}
\end{document}